\documentclass{amsart}
\usepackage{amsmath}
\usepackage{amsfonts}
\usepackage{amsthm, upref}
\usepackage{graphicx}
\usepackage{enumerate}
\usepackage[usenames, dvipsnames]{color}
\usepackage{url}
\usepackage{tikz}
\usepackage{array}
\usetikzlibrary{decorations.pathreplacing}
\input xy
\xyoption{all}

\newcommand{\comment}[1]{}
\newcommand{\eq}{\begin{equation}}
\newcommand{\en}{\end{equation}}

\newcommand{\Hess}{{\mathrm{Hess}} \nobreak\hspace{.16667em plus .08333em} }

\newcommand{\osc}{{\mathrm{osc}}}

\begin{document}

\theoremstyle{plain}
\newtheorem{theorem}{Theorem}[section]
\newtheorem{lemma}[theorem]{Lemma}
\newtheorem{proposition}[theorem]{Proposition}
\newtheorem{corollary}[theorem]{Corollary}

\theoremstyle{definition}
\newtheorem{definition}[theorem]{Definition}
\newtheorem{asmp}[theorem]{Assumption}
\newtheorem{notn}[theorem]{Notation}
\newtheorem{problem}[theorem]{Problem}

\theoremstyle{remark}
\newtheorem{remark}[theorem]{Remark}
\newtheorem{example}[theorem]{Example}
\newtheorem{clm}[theorem]{Claim}

\numberwithin{equation}{section}

\title[Optimization of relative arbitrage]{Optimization of relative arbitrage}

\author[T.-K. L. Wong]{Ting-Kam Leonard Wong}
\address{Department of Mathematics\\ University of Washington\\ Seattle, WA 98195}
\email{wongting@uw.edu}



\thanks{The author would like to thank Soumik Pal for his constant guidance and support during the preparation of the paper, Tatiana Toro for helpful discussions about the proof of Theorem \ref{thm:main}, and Jiashan Wang for help with numerical optimization. He also thanks the anonymous referee who spotted an error in the original definition of the support condition and suggested the current definition. The referee's valuable comments improved greatly the presentation of the paper.}

\date{\today}
\begin{abstract}
In stochastic portfolio theory, a relative arbitrage is an equity portfolio which is guaranteed to outperform a benchmark portfolio over a finite horizon. When the market is diverse and sufficiently volatile, and the benchmark is the market or a buy-and-hold portfolio, functionally generated portfolios introduced by Fernholz provide a systematic way of constructing relative arbitrages. In this paper we show that if the market portfolio is replaced by the equal or entropy weighted portfolio among many others, no relative arbitrages can be constructed under the same conditions using functionally generated portfolios. We also introduce and study a shaped-constrained optimization problem for functionally generated portfolios in the spirit of maximum likelihood estimation of a log-concave density.
\end{abstract}

\maketitle

\section{Introduction} \label{sec:intro}
A major aim of stochastic portfolio theory (see \cite{F02} and \cite{FKSurvey} for an introduction) is to uncover relative arbitrage opportunities under minimal and realistic assumptions on the behavior of equity markets. Consider an equity market with $n$ stocks. The market weight $\mu_i(t)$ of stock $i$ at time $t$ is the market capitalization of stock $i$ divided by the total capitalization of the market. The vector $\mu(t) = (\mu_1(t), ..., \mu_n(t))$ of market weights takes value in the open unit simplex $\Delta^{(n)}$ in ${\Bbb R}^n$ defined by
\[
\Delta^{(n)} = \left\{p = (p_1, ..., p_n): p_i > 0, \sum_{i = 1}^n p_i = 1 \right\}.
\]
For each $t$, the portfolio manager chooses a portfolio vector in $\overline{\Delta^{(n)}}$, where $\overline{\Delta^{(n)}}$ is the closure of $\Delta^{(n)}$. Its components represent the proportions of the current capital invested in each of the stocks. We assume that the portfolio is self-financing and all-long, so short selling is prohibited. The {\it market portfolio} is the portfolio whose portfolio weight at time $t$ is $\mu(t)$. It is a buy-and-hold portfolio since no trading is required after its installment. In general trading is required to maintain the target portfolio weights. A {\it relative arbitrage} with respect to the market portfolio over the horizon $[0, t_0]$ is a portfolio which is guaranteed to outperform the market portfolio at time $t_0$.

We say that the market is {\it diverse} if $\max_{1 \leq i \leq n} \mu_i(t) \leq 1 - \delta$ for some $\delta > 0$ and for all $t$, or more generally if $\mu(t) \in K$ for all $t$ where $K$ is an appropriate subset of $\Delta^{(n)}$. The market is {\it sufficiently volatile} if the cumulated volatility of the market weight grows to infinity in a suitable sense. Assuming the market is diverse and sufficiently volatile, it is possible to construct relative arbitrages with respect to the market portfolio over a finite (but possibly long) horizon; see for example \cite{FKSurvey}, \cite{PW13} and the references therein. In fact, it is possible to construct relative arbitrages whose portfolio weights are deterministic functions of the current market weights. In particular, forecasts of expected returns and the covariance matrix are not required. These portfolios, first introduced in \cite{F99}, are said to be {\it functionally generated}. This is in accordance with the observation by many academics and practitioners (see for example \cite{FGH98}, \cite{DGU09} and \cite{PPA12}) that simple portfolio rules such as the equal and diversity weighted portfolios often beat the market over long periods. Intuitively, these portfolios work by capturing market volatility while controlling the maximum drawdown relative to the market portfolio  (the main ideas will be reviewed in Section \ref{sec:Fernholz}). In \cite{PW14} we proved the converse: a relative arbitrage portfolio (more precisely a {\it pseudo-arbitrage}, see below) depending deterministically on the current market weights must be functionally generated. We emphasize that a relative arbitrage portfolio is supposed to perform well {\it for all} possible realizations of the market weight process satisfying diversity and sufficient volatility. This observation is utilized in \cite{PW14} to allow a geometric, pathwise approach without assuming any stochastic model for the market weight process.

\medskip

There are two important questions that are not fully addressed by the existing theory. First, what happens if the market portfolio is replaced by another benchmark? In \cite{S12} the concept of functionally generated portfolio and the key `master equation' (see Lemma \ref{lem:FernholzDecomp} below) are extended to arbitrary benchmark portfolios. However, little is known about the existence of relative arbitrage under general conditions such as diversity and sufficient volatiltiy. For example, can we beat the equal-weighted portfolio by a functionally generated portfolio in a diverse and sufficiently volatile market, in the same way a functionally generated portfolio beats the market portfolio? More generally, does there exist an infinite hierarchy of relative arbitrages? Is there a `maximal portfolio' which cannot be beaten if only diversity and sufficient volatility are assumed?

Second, is there a sound and applicable optimization theory for relative arbitrages and functionally generated portfolios? Such a theory is clearly of great interest and this problem was raised already in Fernholz's monograph \cite[Problems 3.1.7-8]{F02}. To the best of our knowledge limited progress has been made to optimization of functionally generated portfolios. See \cite{PW13} for an attempt in the two asset case and \cite{PW14} for an approach using optimal transport. On the theoretical side, if the market model is given it is sometimes possible to characterize the highest return relative to the market or a given trading strategy that can be achieved using nonanticipative investment rules over a given time horizon. See \cite{FK10} for the case of Markovian markets, \cite{FK11} for a more general setting which allows uncertainty regarding the drift and diffusion coefficients, and \cite{R11} which  expresses optimal relative arbitrages with respect to Markovian trading strategies as delta hedges. For optimization of functionally generated portfolios, a major difficulty is that the class of functionally generated portfolios is a function space and the optimization has to be nonparametric. Ideally, given historical data or a stochastic model of the market weight process, we want to pick an optimal functionally generated portfolio subject to appropriate constraints.

\medskip

The present paper attempts to give answers to both questions. In this paper we interpret relative arbitrage by what we call {\it pseudo-arbitrage} in \cite{PW13}. This is a model-free concept and the precise definition will be stated in Section \ref{sec:prelim}. We only consider portfolios which are deterministic functions of the current market weight, so a portfolio is represented by a map $\pi: \Delta^{(n)} \rightarrow \overline{\Delta^{(n)}}$. This means that the portfolio manager always chooses $\pi(p)$ when the current market weight is $\mu(t) = p \in \Delta^{(n)}$, regardless of previous price movements. Following \cite{PW14}, in this paper time is discrete and the market is represented by a deterministic sequence $\{\mu(t)\}_{t = 0}^{\infty}$ with state space $\Delta^{(n)}$. No underlying probability space is required.

Regarding the hierarchy of relative arbitrages, we first define a partial order among portfolios. If $\pi$ is a portfolio, we let $V_{\pi}(t)$ be the ratio of the growth of $\$1$ invested in the portfolio to that of $\$1$ invested in the market portfolio, and call it the {\it relative value process}. Let $\pi, \tau: \Delta^{(n)} \rightarrow \overline{\Delta^{(n)}}$ be portfolios. We say that $\tau$ {\it dominates $\pi$ on compacts} (written $\tau \succeq \pi$) if for any compact set $K \subset \Delta^{(n)}$, there exists a constant $\varepsilon = \varepsilon(\pi, \tau, K) > 0$ such that $V_{\tau}(t) / V_{\pi}(t) \geq \varepsilon$ for all $t$ and for all sequences of market weight $\{\mu(t)\}_{t = 0}^{\infty}$ taking values in $K$. That is, the maximum drawdown of $\tau$ relative to $\pi$ is uniformly bounded regardless of the market movement in that region. Since the compact set $K$ is arbitrary, this is a global property and defines a partial order among portfolios. If ${\mathcal S}$ is a family of portfolios, we say that a portfolio $\pi \in {\mathcal S}$ is {\it maximal} in ${\mathcal S}$ if there is no portfolio, other than $\pi$ itself, which dominates $\pi$ on compacts, i.e., $\tau \in {\mathcal S}$ and $\tau \succeq \pi$ implies $\tau = \pi$. In Section \ref{subsec:pseudo} we will relate this partial order with pseudo-arbitrage. Here we note that if $\tau$ is a relative or pseudo-arbitrage with respect to $\pi$ in all diverse and sufficiently volatile markets, it is necessarily the case that $\tau$ dominates $\pi$ on compacts.

Let $\pi: \Delta^{(n)} \rightarrow \overline{\Delta^{(n)}}$ be a portfolio and $\Phi$ be a positive concave function on $\Delta^{(n)}$. We say that $\pi$ is {\it functionally generated} with generating function $\Phi$ if for all $p \in \Delta^{(n)}$,  the vector of coordinatewise ratios $\pi(p)/p$ defines a supergradient of the concave function $\log \Phi$ at $p$ (see Definition \ref{def:fgp} below for the rigorous definition). If $\Phi$ is $C^2$ (twice continuously differentiable), then $\pi$ is necessarily given by
\begin{equation} \label{eqn:fgweight}
\pi_i(p) = p_i \left(1 + D_{e(i) - p} \log \Phi(p) \right), \quad i = 1, ..., n, \quad p \in \Delta^{(n)}.
\end{equation}
Here $D_{e(i) - p}$ is the directional derivative in the direction $e(i) - p$, where $e(i)$ is the vertex of $\overline{\Delta^{(n)}}$ in the $i$-th direction. For example, the market portfolio is generated by the constant function $\Phi(p) \equiv 1$. We say that $\Phi$ is a {\it measure of diversity} if it is $C^2$ and symmetric (invariant under permutations of the coordinates). Let $\overline{e} = \left(\frac{1}{n}, ..., \frac{1}{n}\right)$ be the barycenter of $\Delta^{(n)}$. For portfolios that are continuously differentiable, the following theorem gives a sufficient condition for a portfolio to be maximal.

\begin{theorem} \label{thm:main}
Let $\pi$ be a portfolio generated by a measure of diversity $\Phi$. If
\begin{equation} \label{eqn:integralcondition}
\int_0^1 \frac{1}{\Phi(te(1) + (1 - t)\overline{e})^2} \mathrm{d}t = \infty,
\end{equation}
then $\pi$ is maximal in the class of portfolios $\tau: \Delta^{(n)} \rightarrow \overline{\Delta^{(n)}}$ that are continuously differentiable.
\end{theorem}

This sufficient condition is satisfied by the equal and entropy weighted portfolios (see Table \ref{tab:benchmark} in Section \ref{sec:benchmark} for the definitions) among many others. For the market portfolio the generating function is constant and so the integral in \eqref{eqn:integralcondition} converges. In Section \ref{sec:benchmark} we will show if $\pi$ is functionally generated and $\tau$ dominates $\pi$ on compacts, then $\tau$ must be functionally generated. Thus we may rephrase Theorem \ref{thm:main} by saying that if \eqref{eqn:integralcondition} holds then $\pi$ is maximal in the family of functionally generated portfolios with $C^2$ generating functions. A consequence of Theorem \ref{thm:main} is the following.

\begin{corollary} \label{cor:main}
Under the setting of Theorem \ref{thm:main}, suppose $\tau$ is a $C^1$ portfolio not equal to $\pi$. Then there is a compact set $K \subset \Delta^{(n)}$ and a market weight sequence $\{\mu(t)\}_{t \geq 0}$ taking values in $K$, such that the portfolio value of $\tau$ relative to $\pi$ tends to zero as $t$ tends to infinity.
\end{corollary}

One can interpret Corollary \ref{cor:main} by saying that if $\pi$ is maximal and $\tau \neq \pi$, it is possible to find a diverse and sufficiently volatile market in which $\pi$ beats $\tau$ in the long run. In this sense, for a portfolio $\pi$ satisfying \eqref{eqn:integralcondition}, it is impossible to find a (deterministic) portfolio which is a relative arbitrage with respect to $\pi$ in all diverse and sufficiently volatile markets. Theorem \ref{thm:main} will be proved by comparing the relative concavities of portfolio generating functions.

\medskip

Regarding optimization of functionally generated portfolios, we formulate a {\it shape-constrained optimization problem} in the spirit of maximum likelihood estimation of a log-concave density. For the statistical theory we refer the reader to \cite{DR09}, \cite{CSS10}, \cite{CS10}, \cite{KM10} and \cite{SW10}. Following \cite{PW14}, we associate to each functionally generated portfolio an {\it L-divergence functional} $T\left(\cdot \mid \cdot\right)$ defined on $\Delta^{(n)} \times \Delta^{(n)}$ (see Definition \ref{def:discreteenergy}). Intuitively, $T\left(q \mid p\right)$ measures the potential profit from volatility captured when the market weight jumps from $p$ to $q$ in $\Delta^{(n)}$. Let ${\Bbb P}$ be an intensity measure over the jumps $(p, q)$ which can be defined in terms of data or a given model (examples will be given in Section \ref{sec:optimization}). We maximize 
\[
\int T\left( q \mid p \right) \mathrm{d}{\Bbb P}
\]
over all functionally generated portfolios with or without constraints. This optimization problem is shape-constrained because the generating function of a functionally generated portfolio is concave. We prove that the optimization problem is well-posed and is in a suitable sense {\it consistent} when interpreted as a statistical estimation problem. In this paper we implement this optimization for the case of two assets (analogous to univariate density estimation) and a general algorithm will be the topic of future research. We illustrate a typical application in portfolio management with a case study.

\medskip

The paper is organized as follows. In Section \ref{sec:prelim} we set up the notations and recall the definitions of pseudo-arbitrage and functionally generated portfolio. In Section \ref{sec:benchmark} we extend the framework of \cite{PW14} to benchmark portfolios that are functionally generated. Using a relative concavity lemma given in \cite{CDO07}, we prove Theorem \ref{thm:main} and Corollary \ref{cor:main} in Section \ref{sec:concavity}. Optimization of functionally generated portfolios is studied in Section \ref{sec:optimization} and an empirical case study is presented in Section \ref{sec:empirical}. Several proofs of a more technical nature are gathered in Appendex \ref{sec:appendix}.

\section{Pseudo-arbitrage and functionally generated portfolio} \label{sec:prelim}
\subsection{Portfolio and pseudo-arbitrage}  \label{subsec:pseudo}
We work under the discrete time, deterministic set-up of \cite{PW14} which we briefly recall here. Let $n \geq 2$ be the number of stocks or assets in the market. We endow the open unit simplex $\Delta^{(n)}$ with the Euclidean metric. The open ball in $\Delta^{(n)}$ centered at $p$ with radius $\delta$ is denoted by $B(p, \delta)$. A tangent vector of $\Delta^{(n)}$ is a vector $v = (v_1, ..., v_n) \in {\Bbb R}^n$ satisfying $\sum_{i = 1}^n v_i = 0$. We denote the vector space of tangent vectors of $\Delta^{(n)}$ by $T\Delta^{(n)}$. For $i = 1, ..., n$, we let $e(i) = (0, ..., 0, 1, 0, ..., 0)$ be the vertex of $\Delta^{(n)}$ in the $i$-th direction. If $a$ and $b$ are vectors in ${\Bbb R}^n$, we let $\langle a, b \rangle$ be the Euclidean inner product. The Euclidean norm is denoted by $\|\cdot\|$. If $b$ has nonzero entries, $a / b$ is the vector of the componentwise ratios $a_i / b_i$.

Throughout this paper time is discrete ($t = 0, 1, 2, ...$). Extensions to continuous time will be discussed briefly in Section \ref{sec:continuoustime}. Let $X_i(t) > 0$ be the market capitalization of stock $i$ at time $t$. The total capitalization of the market is then $X_1(t) + \cdots + X_n(t)$. The market weight of stock $i$ is defined by
\[
\mu_i(t) = \frac{X_i(t)}{X_1(t) + \cdots + X_n(t)}, \quad i = 1, ..., n.
\] 
The vector $\mu(t) = (\mu_1(t), ..., \mu_n(t))$ takes values in $\Delta^{(n)}$ and represents the relative sizes of the firms. As the stock prices move the market weights fluctuate accordingly.

As in \cite{PW14}, the stock market is modeled as a deterministic sequence $\{\mu(t)\}_{t \geq 0}$ taking values in $\Delta^{(n)}$, so an underlying probability space is not required. Our approach is analogous to that of universal prediction (see for example \cite{CL06}) where it is not assumed that the data is generated by a stochastic model. Only structural properties such as diversity and sufficient volatility will be imposed on the sequences.

We consider a small investor in this market who cares about the value of his or her portfolio relative to that of the entire market. We restrict ourselves to portfolios which are deterministic functions of the current market weights. Short sales are not allowed and we assume there is no transaction cost.

\begin{definition}[Portfolio and relative value process]
A portfolio is a Borel measurable map $\pi: \Delta^{(n)} \rightarrow \overline{\Delta^{(n)}}$. The market portfolio $\mu$ is the identity map $p \mapsto p$ and we do not distinguish it from the market weight process $\{\mu(t)\}$. Given a portfolio $\pi$, its relative value process $\{V_{\pi}(t)\}_{t \geq 0}$ is defined by $V_{\pi}(0) = 1$ and
\begin{equation} \label{eqn:relativevalue}
\frac{V_{\pi}(t+1)}{V_{\pi}(t)} = 1 + \left\langle \frac{\pi(\mu(t))}{\mu(t)}, \mu(t + 1) - \mu(t) \right\rangle, \quad t \geq 0.
\end{equation}
The weight ratio of the portfolio at $p \in \Delta^{(n)}$ is the vector $\frac{\pi(p)}{p} = \left(\frac{\pi_1(p)}{p_1}, ..., \frac{\pi_n(p)}{p_n}\right)$.
\end{definition}

The relative value $V_{\pi}(t)$ can be interpreted as the ratio of the growth of $\$1$ invested in the portfolio to that of $\$1$ invested in the market portfolio. If $V_{\pi}(t_1) > V_{\pi}(t_0)$, the portfolio outperforms the market portfolio over the (discrete) time interval $[t_0, t_1]$. As mentioned in \cite{PW14}, it is helpful to think of the weight ratio $p \mapsto \frac{\pi(p)}{p}$ as a vector field on $\Delta^{(n)}$. From \eqref{eqn:relativevalue}, the portfolio outperforms the market over $[t, t + 1]$ if the inner product between the displacement $\mu(t + 1) - \mu(t)$ of the market weight and the weight ratio is positive. This means on average the portfolio puts more weight on the assets which perform well relative to the rest of the market. 

\medskip

In the first part of the paper we will study the hierarchy of portfolios defined by the relation `domination on compacts'.

\begin{definition}[Domination on compacts] \label{def:pseudoarbitrage}
Let $\pi$ and $\tau$ be portfolios. We say that $\tau$ dominates $\pi$ on compacts (written $\tau \succeq \pi$) if for any compact subset $K$ of $\Delta^{(n)}$, there exists a constant $C = C(\pi, \tau, K) \geq 0$ such that for any path $\{\mu(t)\}_{t \geq 0} \subset K$, we have
\begin{equation} \label{eqn:lowerbound}
\log\frac{V_{\tau}(t)}{V_{\pi}(t)} \geq -C, \quad t \geq 0.
\end{equation}
\end{definition}

Thus, if $\tau \succeq \pi$, the value of $\pi$ cannot grow at a rate faster than that of $\tau$ under the diversity condition $\mu(t) \in K$, for any compact subset $K$. The relation $\tau \succeq \pi$ defines a partial order among the class of portfolio maps. We include the logarithm in \eqref{eqn:lowerbound} as this formulation is more convenient when we discuss functionally generated portfolios. This definition is closely related to that of pseudo-arbitrage introduced in \cite{PW14}. The definition given below is extended slightly to allow for an arbitrary benchmark portfolio.

\begin{definition}[Pseudo-arbitrage] \label{def:pseudoarbitrage2}
Let $\pi$ and $\tau$ be portfolios, and $K$ be a subset of $\Delta^{(n)}$, not necessarily compact. We say that $\tau$ is a pseudo-arbitrage with respect to $\pi$ on $K$ if the following properties hold:
\begin{enumerate}
\item[(i)] There exists a constant $C = C(\pi, \tau, K) \geq 0$ such that \eqref{eqn:lowerbound} holds for any sequence $\{\mu(t)\}_{t \geq 0} \subset K$.
\item[(ii)] There exists a sequence $\{\mu(t)\}_{t \geq 0} \subset K$ along which $\lim_{t \rightarrow \infty} \log V(t) = \infty$.
\end{enumerate}
\end{definition}

We refer the reader to \cite{PW14} for more discussion of the definition. Here we note that the requirement $\{\mu(t)\}_{t \geq 0} \subset K$ in (i) is a diversity condition which is portfolio-specific, and (ii) refers to the presence of sufficient volatility. The following is an easy consequence of the definitions.

\begin{lemma} \label{lem:easy}
Let $\pi$ and $\tau$ be portfolios. Suppose $\tau$ is a pseudo-arbitrage relative to $\pi$ on $K_j$ for all $j$, where $\{K_j\}$ is a compact exhaustion of $\Delta^{(n)}$. Then $\tau$ dominates $\pi$ on compacts.
\end{lemma}

\begin{definition}[Maximal portfolio]
Let ${\mathcal S}$ be a family of portfolios and $\pi \in {\mathcal S}$. We say that $\pi$ is maximal in ${\mathcal S}$ if there is no portfolio in ${\mathcal S}$, other than $\pi$ itself, which dominates $\pi$ on compacts.
\end{definition}

Note that a maximal portfolio may not exist and may not be unique in the given class. In Section \ref{sec:concavity} we will study the maximal portfolios where ${\mathcal S}$ is the class of portfolios with $C^2$ generating functions. By Lemma \ref{lem:easy}, if $\pi$ is maximal there is no portfolio which is a pseudo-arbitrage with respect to $\pi$ on all sufficiently large compact subsets of $\Delta^{(n)}$. In this sense a maximal portfolio is one which is impossible to beat assuming only diversity and sufficient volatility.

\begin{remark}
The relation `domination on compacts' refers to global properties of portfolios. Even if $\pi$ is maximal, for a {\it fixed} subset $K \subset \Delta^{(n)}$ it may be possible to find a portfolio $\tau$ (depending on $K$) which beats $\pi$ in the long run whenever $\{\mu(t)\} \subset K$. For example, when $n = 2$, it can be shown that the entropy-weighted portfolio beats the equal-weighted portfolio in the long run if $\{\mu(t)\}$ is sufficiently volatile and stays in a certain neighborhood of $\left(\frac{1}{2}, \frac{1}{2}\right)$. This, however, requires that $K$ is known in advance. Maximality of $\pi$ requires that there is no {\it single} $\tau$ which beats $\pi$ on {\it all} compact sets $K \subset \Delta^{(n)}$.
\end{remark}

\subsection{Functionally generated portfolio} \label{sec:fgp}
Functionally generated portfolio was first introduced in a general form in \cite{F99}. We will follow the intrinsic treatment in \cite[Section 2]{PW14} which emphasizes the relationship with convex analysis. Throughout the paper we will rely heavily on results from convex analysis and a standard reference is \cite{R70}.

\begin{definition} [Functionally generated portfolios] \label{def:fgp} {\ }
Let $\pi$ be a portfolio and $\Phi: \Delta^{(n)} \rightarrow (0, \infty)$ be a concave function. We say that $\pi$ is generated by $\Phi$ if the inequality
\begin{equation} \label{eqn:superdiff}
1 + \left\langle \frac{\pi(p)}{p}, q - p \right\rangle \geq \frac{\Phi(q)}{\Phi(p)}
\end{equation}
holds for all $p, q \in \Delta^{(n)}$. We call $\Phi$ the generating function of $\pi$. We denote by ${\mathcal{FG}}$ the collection of all functionally generated portfolios $(\pi, \Phi)$ where $\pi$ is generated by the concave function $\Phi$.
\end{definition}

It is known (see \cite[Proposition 5]{PW14}) that the generating function is unique up to a positive multiplicative constant, so the use of `the' in the above definition is justified (up to the constant). On the other hand, by Lemma \ref{lem:superdiff}(ii) below a non-smooth concave function $\Phi$ generates multiple portfolios but they differ only on the set where $\Phi$ is not differentiable (i.e., the superdifferential $\partial \log \Phi(p)$ has more than one element), and this set has Lebesgue measure zero (relative to $\Delta^{(n)}$) by \cite[Theorem 25.5]{R70}. Note that here the generating function is concave by definition, while in \cite{F02} non-concave generating functions are allowed. See Theorem \ref{thm:PW14} and Proposition \ref{prop:MCMfgp} below for a justification of our definition.

Let $\Phi$ be a concave function on $\Delta^{(n)}$ and $p \in \Delta^{(n)}$. The {\it superdifferential} of $\Phi$ at $p$ is the set $\partial \Phi(p)$ defined by
\begin{equation} \label{eqn:superdiffdef}
\partial \Phi(p) = \{\xi \in T\Delta^{(n)}: \Phi(p) + \langle \xi, q - p \rangle \geq \Phi(q) \ \forall q \in \Delta^{(n)}\}.
\end{equation}
If $\Phi$ is concave and positive, it can be shown that $\log \Phi$ is also a concave function, and
\begin{equation} \label{eqn:supdiffequal}
\partial \log \Phi(p) = \frac{1}{\Phi(p)} \partial \Phi(p) = \left\{\frac{1}{\Phi(p)} \xi: \xi \in \partial \Phi(p)\right\}.
\end{equation}

\begin{lemma}\cite[Proposition 6]{PW14} \label{lem:superdiff} Let $\Phi$ be a positive concave function on $\Delta^{(n)}$.
\begin{enumerate}
\item[(i)] Let $\pi$ be a portfolio generated by $\Phi$. Then for $p \in \Delta^{(n)}$, the tangent vector $v = (v_1, ..., v_n)$ defined by
\begin{equation} \label{eqn:definev}
v_i = \frac{\pi_i(p)}{p_i} - \frac{1}{n} \sum_{j = 1}^n \frac{\pi_j(p)}{p_j}, \quad i = 1, ..., n,
\end{equation}
belongs to $\partial \log \Phi(p)$.
\item[(ii)] Conversely, if $v \in \partial \log\Phi(p)$, then the vector $\pi = (\pi_1, ..., \pi_n)$ defined by
\begin{equation} \label{eqn:definepi}
\frac{\pi_i}{p_i} = v_i + 1 - \sum_{j = 1}^n p_jv_j, \quad i = 1, ..., n,
\end{equation}
is an element of $\overline{\Delta^{(n)}}$. In particular, any measurable selection of $\partial \log \Phi$ (a Borel measurable map $\xi: \Delta^{(n)} \rightarrow T\Delta^{(n)}$ such that $\xi(p) \in \partial \log \Phi(p)$ for all $p \in \Delta^{(n)}$) defines via \eqref{eqn:definepi} a portfolio generated by $\Phi$. (By \cite[Theorem 14.56]{RW98}, there is always a measurable selection of $\partial \log \Phi$.)
\end{enumerate}
Moreover, the operations $\pi \mapsto v$ and $v \mapsto \pi$ defined by \eqref{eqn:definev} and \eqref{eqn:definepi} are inverses of each other.
\end{lemma}

From \eqref{eqn:definepi}, it can be seen that Fernholz's definition (see \cite[Theorem 3.1.5]{F02}) is consistent with ours. If $\pi$ is generated by $\Phi$, the weight ratio vector field $\frac{\pi}{p}$ is {\it conservative} on $\Delta^{(n)}$ and its potential function is given by the logarithm of the generating function $\Phi$. Here is a precise statement and the details can be found in the proof of \cite[Theorem 8]{PW14}. Let $\pi$ be a portfolio. If $\gamma: [0, 1] \rightarrow \Delta^{(n)}$ is a piecewise linear path in $\Delta^{(n)}$, we let 
\begin{equation}  \label{eqn:lineintegral}
I_{\pi}(\gamma) := \int_{\gamma} \frac{\pi}{p} \mathrm{d}p \equiv \int_0^1 \sum_{i = 1}^n \frac{\pi_i(\gamma(t))}{p_i(\gamma(t))}\gamma'_i(t)\mathrm{d}t
\end{equation}
be the line integral of the weight ratio along $\gamma$. If $\pi$ is functionally generated, the weight ratio $\frac{\pi}{p}$ is conservative in the sense that this line integral is zero whenever $\gamma$ is closed, i.e., $\gamma(0) = \gamma(1)$. Moreover, for any $p, q \in \Delta^{(n)}$ we have
\begin{equation} \label{eqn:lineintegral2}
\log \Phi(q) - \log \Phi(p) = I_{\pi}(\gamma),
\end{equation}
where $\gamma$ is any piecewise linear path from $p$ to $q$. In classical terminology, $\log \Phi$ is then the potential function of the weight ratio vector field. Fernholz's decomposition (see Lemma \ref{lem:FernholzDecomp} below) shows that the log relative value $\log V_{\pi}(t)$ can be decomposed as the sum of the increment of $\log \Phi(\mu(t))$ and a non-decreasing process related to market volatility.

The concavity of the generating function will be measured in terms of the L-divergence introduced in \cite{PW14}.

\begin{definition}[L-divergence] \label{def:discreteenergy}
Let $\pi$ be a portfolio generated by a concave function $\Phi: \Delta^{(n)} \rightarrow (0, \infty)$. The L-divergence functional of the pair $(\pi, \Phi)$ is the function $T: \Delta^{(n)} \times \Delta^{(n)} \rightarrow [0, \infty)$ defined by
\begin{equation} \label{eqn:discreteenergy}
T\left(q \mid p \right) = \log \left(1 + \left\langle \frac{\pi(p)}{p}, q - p \right\rangle \right) - \log \frac{\Phi(q)}{\Phi(p)}, \quad p, q \in \Delta^{(n)}.
\end{equation}
\end{definition}

Using \eqref{eqn:superdiff}, it can be shown that $T\left(q \mid p \right) \geq 0$ and $T\left(q \mid p \right) = 0$ only if $\Phi$ is affine on the line segment containing $p$ and $q$. $T\left(\cdot \mid \cdot\right)$ is a logarithmic version (hence the `L') of {\it Bergman divergence} used in information geometry (see \cite{AC10}) and should be thought of as a measure of the concavity of $\Phi$.

With these definitions, the main results of \cite{PW14} can be summarized as follow.

\begin{theorem}[Pseudo-arbitrages relative to the market portfolio] \cite[Theorem 1, Theorem 2]{PW14} \label{thm:PW14}
A portfolio $\pi$ is a pseudo-arbitrage relative to the market portfolio $\mu$ on a convex subset $K \subset \Delta^{(n)}$ if and only if $\pi$ is generated by a concave function $\Phi: \Delta^{(n)} \rightarrow (0, \infty)$ which is bounded below on $K$ and $T \left( \cdot \mid \cdot \right)$ is not identically zero on $K \times K$. Moreover, these portfolios correspond to solutions of an optimal transport problem.
\end{theorem}

In Section \ref{sec:concavity} we will focus on functionally generated portfolios with $C^2$ generating functions.

\begin{definition} \label{def:fg} \label{defn:C2fgp} {\ }
\begin{enumerate}
\item[(i)] We denote by ${\mathcal{FG}}^2$ the collection of functionally generated portfolios whose generating functions are $C^2$ and concave. An element of ${\mathcal{FG}}^2$ is denoted by either $\pi$, $\Phi$ or $(\pi, \Phi)$ where $\pi$ is generated by $\Phi$. In this case $\pi$ is necessarily given by \eqref{eqn:fgweight}.
\item[(ii)] A positive $C^2$ concave function $\Phi$ on $\Delta^{(n)}$ is called a measure of diversity if it is symmetric, i.e.,
\[
\Phi(p_1, ..., p_n) = \Phi(p_{\sigma(1)}, ..., p_{\sigma(n)})
\]
for all $p \in \Delta^{(n)}$ and any permutation $\sigma$ of $\{1, ..., n\}$.
\end{enumerate}
\end{definition}

Measure of diversity was introduced by Fernholz in \cite[Section 4]{F99}. Some examples are given in Table \ref{tab:benchmark} and more can be found in \cite[Section 3.4]{F02}. A measure of diversity gives a numerical measure of the concentration of the capital distribution $\mu(t) = \left(\mu_1(t), ..., \mu_n(t)\right)$ and also generates a portfolio.

\section{Benchmarking a functionally generated portfolio} \label{sec:benchmark}
Fix a portfolio $\pi$ generated by a concave function $\Phi: \Delta^{(n)} \rightarrow (0, \infty)$ and call it the {\it benchmark portfolio}. Some examples we have in mind are given in Table \ref{tab:benchmark}. All of these portfolios are generated by measures of diversity.

\begin{table}
\caption{Examples of functionally generated portfolios}
\label{tab:benchmark}
\begin{tabular}{lll}
\hline\noalign{\smallskip}
    Name & Portfolio weights  & Generating function \\
\noalign{\smallskip}\hline\noalign{\smallskip}
    Market & $\pi_i(p) = p_i$ & $\Phi(p) = 1$  \\ 
    Diversity-weighted ($0 < r < 1$) & $\pi_i(p) = \frac{p_i^r}{\sum_{j = 1}^n p_j^r}$ & $\Phi(p) = \left( \sum_{j = 1}^n p_j^r \right)^{\frac{1}{r}}$ \\ 
    Equal-weighted & $\pi_i(p) = \frac{1}{n}$ & $\Phi(p) = \left(p_1 p_2 \cdots p_n \right)^{\frac{1}{n}}$  \\ 
    Entropy-weighted & $\pi_i(p) = \frac{-p_i \log p_i}{\sum_{j = 1}^n -p_j \log p_j}$ & $\Phi(p) = \sum_{j = 1}^n -p_j \log p_j$  \\ 
\noalign{\smallskip}\hline
\end{tabular}
\end{table}

As mentioned in the introduction, it can be proved that many functionally generated portfolios (including the three nontrivial examples above) outperform the market over sufficiently long periods under the assumptions of diversity and sufficient volatility. As these hypotheses appear to hold empirically, many functionally generated portfolios outperform the market over long periods. See \cite[Chapter 6]{F02} for several case studies using data of the US stock market. Since these portfolios contain no proprietary modeling, behave reasonably well and are easily replicable, they also serve as alternative benchmarks as discussed in practitioner papers such as \cite{FGH98} and \cite{HCKL11}. It is natural to ask whether we can construct relative or pseudo-arbitrages with respect to these portfolios.

\subsection{Fernholz's decomposition} \label{sec:Fernholz}
The relative value process of a functionally generated portfolio satisfies an elegant decomposition formula. It is a direct consequence of \eqref{eqn:discreteenergy} and \eqref{eqn:relativevalue} and can be motivated by the vector field interpretation discussed in Section \ref{sec:fgp}. 

\begin{lemma}[Fernholz's decomposition] \label{lem:FernholzDecomp} \cite[Theorem 3.1]{F99} \cite[Lemma 7]{PW14}
If $\pi$ is generated by a concave function $\Phi$, the relative value process $V_{\pi}$ has the decomposition
\begin{equation} \label{eqn:FernholzDecomp}
\log V_{\pi}(t) = \log \frac{\Phi(\mu(t))}{\Phi(\mu(0))} + A(t),
\end{equation}
where $A(t) = \sum_{k = 0}^{t-1} T\left(\mu(k+1) \mid \mu(k)\right)$ is non-decreasing. We call $A(t)$ the drift process of the portfolio.
\end{lemma}

\begin{figure}
\includegraphics[scale=1]{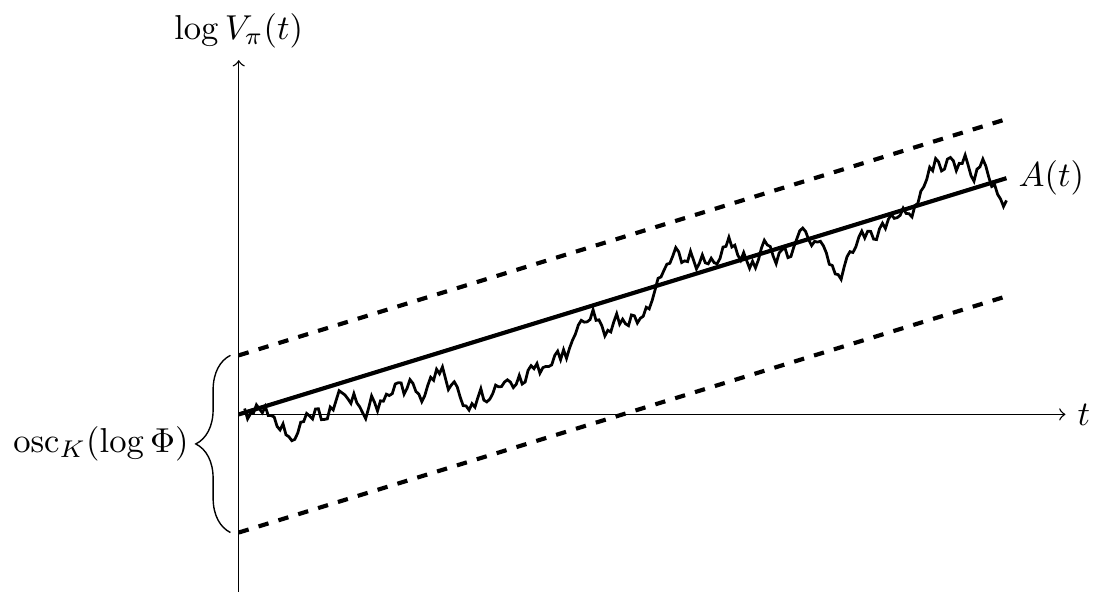}
\caption{Hypothetical performance of a functionally generated portfolio. If the market weight $\mu(t)$ stays within a subset $K \subset \Delta^{(n)}$, the relative value process will stay within the dashed curves which are vertical translations of the drift process $A(t)$. The width of the `sausage' is given by the oscillation of $\log \Phi$ on $K$ defined by $\osc_K(\log \Phi) = \sup_{p, q \in K} |\log \Phi(q) - \log \Phi(p)|$.} 
\label{fig:FernholzDecomp}
\end{figure}

The key idea of the decomposition is that over any period $[t_0, t_1]$ where $\log \Phi(\mu(t_1))$ and $\log \Phi(\mu(t_0))$ are approximately equal, the portfolio will outperform the market by an amount equal to $A(t_1) - A(t_0)$, see Figure \ref{fig:FernholzDecomp} for an illustration. For this reason, the drift process $A(t)$ can be thought of as the cumulative amount of market volatility captured by the portfolio. The condition of sufficient volatility requires that $A(t)$ grows unbounded as $t \rightarrow \infty$. Empirical studies (see for example \cite[Figure 11.2]{FKSurvey}) show that $A$ increases at a roughly linear rate depending on the portfolio and market volatility. Thus, as long as the fluctuation of $\log \Phi(\mu(t))$ remains bounded, the drift process will dominate in the long run and the portfolio will outperform the market. The assumption on diversity is imposed to bound $\log \Phi(\mu(t))$. For (say) the entropy-weighted portfolio, $\log \Phi(\mu(t))$ is bounded as long as $\max_{1 \leq i \leq n} \mu_i(t) \leq 1 - \delta$ for some $\delta > 0$, so we can take $K$ in Definition \ref{def:pseudoarbitrage2} and Theorem \ref{thm:PW14} to be the set $\{p \in \Delta^{(n)}: \max_{1 \leq i \leq n} p_i \leq 1 - \delta\}$ (this is the definition of diversity stated in \cite{F99} and \cite{FKSurvey}). For other portfolios such as the equal-weighted portfolio, this condition is not enough and we require that $\mu(t)$ stays within a compact subset of $\Delta^{(n)}$. Thus the set $K$ is portfolio-specific. Fernholz's decomposition is implemented in the \verb"R" package \verb"RelValAnalysis" (available on \verb"CRAN") written by the author.

\subsection{Domination on compacts}
In \cite{PW14} pseudo-arbitrages with respect to the market portfolio are characterized in terms of a property called {\it multiplicative cyclical monotonicity} (MCM). It is a variant of cyclical monotonicity in convex analysis (see \cite[Section 24]{R70}) and is equivalent to $c$-cyclical monotonicity in optimal transport for a special cost function. Intuitively, this property requires that the portfolio outperforms the market portfolio whenever the market weight goes through a cycle. It is natural to extend the definition as follow.

\begin{definition}[Relative multiplicative cyclical monotonicity - RMCM]
Let $\pi$ and $\tau$ be portfolios. We say that $\tau$ satisfies multiplicative cyclical monotonicity relative to $\pi$ if over any discrete cycle
\[
\mu(0), \mu(1), ..., \mu(m), \mu(m+1) = \mu(0)
\]
in $\Delta^{(n)}$, we have
\begin{equation} \label{eqn:rmcm}
V_{\tau}(m + 1) \geq V_{\pi}(m + 1).
\end{equation}
\end{definition}

In \cite{PW14} we proved that functionally generated portfolios are characterized by the MCM property relative to the market portfolio.

\begin{proposition} \cite[Proposition 4]{PW14} \label{prop:MCMfgp}
A portfolio satisfies MCM relative to the market portfolio if and only if it is generated by a positive concave function.
\end{proposition}

For an arbitrary functionally generated benchmark portfolio, we can generalize Proposition \ref{prop:MCMfgp} as follow. This result provides equivalent formulations of the partial order $\succeq$ that are easier to work with. The proof is analogous to those of Proposition 4 and Theorem 1 of \cite{PW14}. 

\begin{theorem} \label{prop:MCM}
Let $\pi$ be a portfolio generated by a concave function $\Phi: \Delta^{(n)} \rightarrow (0, \infty)$, and let $\tau$ be a portfolio. The following statements are equivalent.
\begin{enumerate}
\item[(i)] $\tau$ dominates $\pi$ on compacts, i.e., $\tau \succeq \pi$.
\item[(ii)] $\tau$ satisfies MCM relative to $\pi$.
\item[(iii)] $\tau$ is generated by a concave function $\Psi$, and the L-divergence $T_{\tau}\left(\cdot\mid \cdot\right)$ of $(\tau, \Psi)$ dominates $T_{\pi}\left(\cdot\mid \cdot\right)$ of $(\pi, \Phi)$ in the sense that
\begin{equation} \label{eq:divergenceineq}
T_{\tau}\left(q\mid p\right) \geq T_{\pi}\left(q\mid p\right)
\end{equation}
for all $p, q \in \Delta^{(n)}$.
\end{enumerate}
\end{theorem}
\begin{proof}
(i) $\Rightarrow$ (ii): Suppose $\tau$ dominates $\pi$ on compacts. If $\tau$ does not satisfy MCM relative to $\pi$, we can find a discrete cycle $\{\mu(t)\}_{t = 0}^{m + 1}$ such that $\eta := V_{\tau}(m + 1) / V_{\pi}(m + 1) < 1$. Consider the market weight sequence which goes over this cycle again and again, i.e., $\mu(t) = \mu(t + (m+1))$ for all $t$. Then
\[
\frac{V_{\tau}(k(m + 1))}{V_{\pi}(k(m + 1))} = \eta^k
\]
for all $k \geq 0$ and the ratio tends to $0$ as $k \rightarrow \infty$. This contradicts the hypothesis $\tau \succeq \pi$. Thus if $\tau$ dominates $\pi$ on compacts then $\tau$ satisfies MCM relative to $\pi$. 

\medskip
(ii) $\Rightarrow$ (iii): Suppose $\tau$ satisfies MCM relative to $\pi$. Since $V_{\mu}(\cdot) \equiv 1$ and $\pi$ satisfies MCM relative to the market portfolio (by Proposition \ref{prop:MCMfgp}), $\tau$ satisfies MCM relative to the market portfolio as well. By Proposition \ref{prop:MCMfgp} again $\tau$ has a generating function $\Psi$. To prove \eqref{eq:divergenceineq}, let $p, q \in \Delta^{(n)}$ with $p \neq q$. Let $\{q = \mu(1), ..., \mu(m), \mu(m+1) = p\}$ be a partition of the line segment $[q, p]$. Then if $\mu(0) = p$, $\{\mu(k)\}_{k = 0}^{m+1}$ is a cycle which starts at $p$, jumps to $q$ and then returns to $p$ along the partition. Then the RMCM inequality \eqref{eqn:rmcm} implies
\begin{equation} \label{eq:MCM}
\begin{split}
& \left( 1 + \left\langle \frac{\tau(p)}{p}, q - p \right\rangle \right) \prod_{k = 1}^m \left( 1 + \left\langle \frac{\tau(\mu(k))}{\mu(k)}, \mu(k+1) - \mu(k) \right\rangle \right) \\
& \geq \left( 1 + \left\langle \frac{\pi(p)}{p}, q - p \right\rangle \right) \prod_{k = 1}^m \left( 1 + \left\langle \frac{\pi(\mu(k))}{\mu(k)}, \mu(k+1) - \mu(k) \right\rangle \right).
\end{split}
\end{equation}
Taking log on both sides, we have
\begin{equation*}
\begin{split}
& \log\left( 1 + \left\langle \frac{\tau(p)}{p}, q - p \right\rangle \right) + \sum_{k = 1}^m \log \left( 1 + \left\langle \frac{\tau(\mu(k))}{\mu(k)}, \mu(k+1) - \mu(k) \right\rangle \right) \\
& \geq \log \left( 1 + \left\langle \frac{\pi(p)}{p}, q - p \right\rangle \right) + \sum_{k = 1}^m \log \left( 1 + \left\langle \frac{\pi(\mu(k))}{\mu(k)}, \mu(k+1) - \mu(k) \right\rangle \right).
\end{split}
\end{equation*}
By the fundamental theorem of calculus for concave function and Taylor approximation, we can choose a sequence of partitions with mesh size going to zero, along which
\begin{equation*}
\begin{split}
\sum_{k = 1}^m \log \left( 1 + \left\langle \frac{\pi(\mu(k))}{\mu(k)}, \mu(k+1) - \mu(k) \right\rangle \right) &\rightarrow \int_{\gamma} \frac{\pi}{\mu} \mathrm{d}\mu = \log \frac{\Phi(p)}{\Phi(q)},\\
\sum_{k = 1}^m \log \left( 1 + \left\langle \frac{\tau(\mu(k))}{\mu(k)}, \mu(k+1) - \mu(k) \right\rangle \right)&\rightarrow \int_{\gamma} \frac{\tau}{\mu} \mathrm{d}\mu = \log \frac{\Psi(p)}{\Psi(q)},
\end{split}
\end{equation*}
where $\gamma$ is the line segment from $q$ to $p$. Taking the corresponding limit in \eqref{eq:MCM}, we obtain the desired inequality \eqref{eq:divergenceineq}. 

\medskip
(iii) $\Rightarrow$ (i): Let $\{\mu(t)\}_{t \geq 0}$ be any market weight sequence. By Lemma \ref{lem:FernholzDecomp} we can write
\[
\log \frac{V_{\tau}(t)}{V_{\pi}(t)} = \log \frac{\Psi(\mu(t)) / \Psi(\mu(0))}{\Phi(\mu(t)) / \Phi(\mu(0))} + \left(A_{\tau}(t) - A_{\pi}(t)\right),
\]
where $A_{\tau}$ and $A_{\pi}$ are the drift processes of $\tau$ and $\pi$ respectively. By (iii), $A_{\tau}(t) - A_{\pi}(t)$ is non-decreasing in $t$. Since $\log \frac{\Psi(\mu(t)) / \Psi(\mu(0))}{\Phi(\mu(t)) / \Phi(\mu(0))}$ is bounded as long as $\mu(t)$ stays within a compact subset of $\Delta^{(n)}$, $\tau$ dominates $\pi$ on compacts.
\qed\end{proof}

Theorem \ref{prop:MCM} reduces the study of the partial order $\tau \succeq \pi$ to comparing the relative concavities of generating functions, where concavity is measured by the $L$-divergence. In this paper we focus on generating functions that are twice continuously differentiable. Then the infinitesimal version of \eqref{eq:divergenceineq} leads to second order differential inequalities. 

\begin{definition}[Drift quadratic form] \label{def:driftform}
Let $(\pi, \Phi) \in {\mathcal{FG}}^2$. Its drift quadratic form, denoted by both $H_{\pi}$ and $H_{\Phi}$, is defined by
\[
H_{\pi}(p)(v, v) := \frac{-1}{2\Phi(p)} \Hess \Phi(p)(v, v), \quad p \in \Delta^{(n)}, v \in T\Delta^{(n)}.
\]
Here $\Hess \Phi$ is the Hessian of $\Phi$ regarded as a quadratic form. By definition, it is given by
\begin{equation}  \label{eqn:Hessian}
\Hess \Phi(p)(v, v) = \left.\frac{\mathrm{d}^2}{\mathrm{d}t^2} \Phi(p + tv)  \right|_{t = 0}.
\end{equation}
\end{definition}

\begin{lemma} \label{lem:drift}
Let $(\pi, \Phi), (\tau, \Psi) \in {\mathcal{FG}}^2$, and let $T_{\pi}$ and $T_{\tau}$ be their corresponding L-divergences. If $\tau \succeq \pi$ and therefore $T_{\tau}\left(q \mid p\right) \geq T_{\pi}\left(q \mid p\right)$ for all $p, q \in \Delta^{(n)}$, then $H_{\tau} \geq H_{\pi}$ in the sense that
\begin{equation} \label{eqn:driftineq}
H_{\tau}(p)(v, v) \geq  H_{\pi}(p)(v, v)
\end{equation}
for all $p \in \Delta^{(n)}$ and $v \in T\Delta^{(n)}$.
\end{lemma}
\begin{proof}
The lemma follows immediately from the Taylor approximation
\begin{equation} \label{eqn:taylor}
T_{\pi} \left(p + tv \mid p \right) = \frac{-1}{2\Phi(p)}\Hess \Phi(p)(tv, tv) + o\left(t^2\right).
\end{equation}
where $p \in \Delta^{(n)}$, $v$ is a tangent vector, and $t \in {\Bbb R}$ is small.
\qed\end{proof}

As a consequence of Lemma \ref{lem:drift}, in order to show that a portfolio $\pi \in {\mathcal{FG}}^2$ is maximal in ${\mathcal{FG}}^2$, it is enough to show that its drift quadratic form $H_{\pi}$ is not dominated (in the sense of \eqref{eqn:driftineq}) by that of some other portfolio. This is the approach we use in Section \ref{sec:concavity} to prove Theorem \ref{thm:main}. Simple examples show, however, that $H_{\tau} \geq H_{\pi}$ does not imply $T_{\tau} \geq T_{\pi}$.

\begin{example}[Diversity-weighted portfolio]
For $0 < r < 1$, the diversity-weighted portfolio $\pi$ introduced at the beginning of this section is generated by the function
\[
\Phi(p) = \left(\sum_{j = 1}^n p_j^r\right)^{\frac{1}{r}}.
\]
It is easy to show that $\Phi$ is bounded below by $1$. Let $\tau$ be the portfolio generated by $\Psi := \Phi - 1$. Then it can be shown that $\tau \succeq \pi$. To see this, write the L-divergence \eqref{eqn:discreteenergy} in the form
\begin{equation} \label{eqn:divergence2}
T_{\pi}\left(q \mid p\right) = \log \frac{\Phi(p) + D_{q - p}\Phi(p)}{\Phi(q)}, \quad p, q \in \Delta^{(n)}.
\end{equation}
Then
\[
T_{\tau}\left(q \mid p\right) = \log \frac{\left(\Phi(p) - 1\right) + D_{q - p}\Phi(p)}{\Phi(q) - 1} \geq T_{\pi}\left(q \mid p\right).
\]
From \eqref{eqn:divergence2}, we can show that for a portfolio $(\pi, \Phi)$ to be maximal in ${\mathcal{FG}}^2$, it is necessary that the continuous extension of $\Phi$ to the closure $\overline{\Delta^{(n)}}$ (which exists by \cite[Theorem 10.3]{R70}) vanishes at all the vertices $e(1)$, ..., $e(n)$ (because otherwise we can subtract an affine function from $\Phi$ and make $T$ larger). However this condition is not sufficient for $\pi$ to be maximal in ${\mathcal{FG}}^2$.
\end{example}

\section{Relative concavity and maximal portfolios} \label{sec:concavity}
\subsection{Two asset case}
In this section we study the maximal portfolios in ${\mathcal{FG}}^2$ and prove Theorem \ref{thm:main}. To illustrate the ideas involved we first give a proof of the maximality of the equal-weighted portfolio for $n = 2$. This result is the starting point of this paper.

\begin{proposition} \label{prop:equalweight}
For $n = 2$, the equal-weighted portfolio $\pi \equiv \left(\frac{1}{2}, \frac{1}{2}\right)$ generated by the geometric mean $\Phi(p) = \sqrt{p_1p_2}$ is maximal in ${\mathcal{FG}}^2$.
\end{proposition}
\begin{proof}
Let $(\tau, \Psi) \in {\mathcal{FG}}^2$ be a portfolio which dominates $(\pi, \Phi)$ on compacts. Define $u(x) = \Phi(x, 1 - x) = \sqrt{x(1 - x)}$ and let $v(x) = \Psi(x, 1 - x)$, $x \in (0, 1)$. Then $u$ and $v$ are positive $C^2$ concave functions on $(0, 1)$. By Theorem \ref{prop:MCM} and Lemma \ref{lem:drift}, the drift quadratic form of $\tau$ dominates that of $\pi$. Using \eqref{eqn:Hessian}, we have the differential inequality
\begin{equation} \label{eqn:n=2domination}
\frac{-v''(x)}{v(x)} \geq \frac{-u''(x)}{u(x)} = \frac{1}{4\left(x(1 - x)\right)^2}, \quad x \in (0, 1).
\end{equation}
We claim that $v$ also generates the equal-weighted portfolio, and so $\tau = \pi$.

We will use a transformation which amounts to a change of num\'{e}raire using $y = \log \frac{x}{1 - x}$. See the binary tree model in \cite[Section 4]{PW13} for the motivation of this transformation and related results. Define a function $\tau_1: (0, 1) \rightarrow [0, 1]$ by
\begin{equation} \label{eqn:n=2weight}
\tau_1(x) = x + x(1 - x) \frac{v'(x)}{v(x)} = x \left[1 + (1 - x) (\log v)'(x)\right].
\end{equation}
By \eqref{eqn:fgweight}, this is the portfolio weight of stock $1$ generated by $v$ and $\tau_1$ takes value in $[0, 1]$. Let $y = \log \frac{x}{1 - x}$, so $x = \frac{e^y}{1 + e^y}$. Define $q: {\Bbb R} \rightarrow [0, 1]$ by
\[
q(y) = \tau_1(x) = \frac{e^y}{1 + e^y} + \frac{e^y}{(1 + e^y)^2} \frac{v'(x)}{v(x)}, \quad x = \frac{e^y}{1 + e^y}, \quad y \in {\Bbb R}.
\]
For the equal-weighted portfolio the corresponding portfolio weight function is identically $\frac{1}{2}$. It follows from a straightforward computation that
\[
q(y)(1 - q(y)) - q'(y) = \frac{-e^{2y}}{(1 + e^y)^4} \frac{v''(x)}{v(x)}.
\]
Now \eqref{eqn:n=2domination} can be rewritten in the form
\begin{equation} \label{eqn:transformeddrift}
q(y)(1 - q(y)) - q'(y) \geq \frac{1}{4}, \quad y \in {\Bbb R}.
\end{equation}
The proof is then completed by the following elementary result.
\qed\end{proof}

\begin{lemma} \label{lem:diffeqn}
Suppose $q: {\Bbb R} \rightarrow [0, 1]$ is differentiable and $q(1 - q) - q' \geq 1/4$ on ${\Bbb R}$. Then $q \equiv 1/2$.
\end{lemma}
\begin{proof}
Since $0 \leq q(y) \leq 1$, we have
\[
q' \leq q(1 - q) - \frac{1}{4} \leq \frac{1}{4} - \frac{1}{4} = 0,
\]
so $q$ is non-increasing. If $q(y_0) = q_0 < \frac{1}{2}$ for some $y_0$, then on $y \in [y_0, \infty]$, $q$ must satisfy the differential inequality
\[
q'(y) \leq q_0(1 - q_0) - \frac{1}{4} < 0,
\]
which contradicts the fact that $q(y) \geq 0$. Similarly, if $q(y_0) = q_0 > \frac{1}{2}$ for some $y_0$, the same inequality is satisfied on $(-\infty, y_0]$, again a contradiction. Thus we get $q(y) \equiv \frac{1}{2}$ for all $y \in {\Bbb R}$.
\qed\end{proof}

The main idea of the proof of Proposition \ref{prop:equalweight} is that for a portfolio to dominate the equal-weighted portfolio $\pi$ on compacts, it must be more aggressive than $\pi$ {\it everywhere} on the simplex. This means buying more and more the underperforming stock at a sufficiently fast rate satisfying \eqref{eqn:transformeddrift}, but this is impossible to continue up to the boundary of the simplex. While there is a multi-dimensional analogue of the differential inequality \eqref{eqn:transformeddrift} (see \cite[Theorem 9]{PW14}), we are unable to extend this proof to the multi-asset case since the market and portfolio weights can move in many directions. Instead, we will work with portfolio generating functions and use the simple but powerful tools of convex analysis.

\subsection{Main result} \label{subsec:relative}
Before we give the proof of Theorem \ref{thm:main} we note that the integral condition \eqref{eqn:integralcondition} is sufficient to capture many important examples. The proof is an exercise in elementary calculus and is left to the reader.

\begin{lemma}
The following portfolios satisfy \eqref{eqn:integralcondition}.
\begin{enumerate}
\item[(i)] The equal-weighted portfolio $\pi \equiv \left(\frac{1}{n}, ..., \frac{1}{n}\right)$ generated by the geometric mean $\Phi(p) = \left(p_1 \cdots p_n\right)^{\frac{1}{n}}$.
\item[(ii)] The entropy-weighted portfolio $\pi_i = -(p_i \log p_i) / \Phi(p)$ generated by the Shannon entropy $\Phi(p) = -\sum_{j = 1}^n p_j \log p_j$.
\end{enumerate}
\end{lemma}

The main ingredient of the proof of Theorem \ref{thm:main} is the following ingenious observation taken from \cite{CDO07} and \cite[Lemma 2]{CDOS09} (it is called the relative convexity lemma in these references). It can be proved by direct differentiation.

\begin{lemma}[Relative concavity lemma] \cite{CDO07} \label{lem:relativeconcavity}
Let $-\infty < a < b \leq \infty$ and $c, C: [a, b) \rightarrow {\Bbb R}$ be continuous. Suppose $u, v: [a, b) \rightarrow (0, \infty)$ are $C^2$ and satisfy the differential equations
\begin{equation*}
\begin{split}
u''(x) + c(x)u(x) &= 0, \quad x \in [a, b), \\
v''(x) + C(x)v(x) &= 0, \quad x \in [a, b).
\end{split}
\end{equation*}
Define $F: [a, b) \rightarrow [0, \infty)$ by
\[
F(x) = \int_a^x \frac{1}{u(t)^2} \mathrm{d}t, \quad x \in [a, b).
\]
Let $G$ be the inverse of $F$ defined on $[0, \ell)$, where $\ell = \lim_{x \uparrow b} F(x)$. Then the function
\[
w(y) := \frac{v(G(y))}{u(G(y))}
\]
defined on $[0, \ell)$ satisfies the differential equation
\[
w''(y) = -(C(x) - c(x))u(x)^4w(y), \quad 0 \leq y < \ell, \quad x = G(y).
\]
In particular, if $C(x) \geq c(x)$ on $[a, b)$, then $w$ is concave on $[0, \ell)$.
\end{lemma}

We also need some convex analytic properties of functionally generated portfolios.

\begin{lemma} \label{lem:FGconvex}
Let $\pi^{(1)}, \pi^{(2)} \in {\mathcal{FG}}$ be generated by $\Phi^{(1)}$ and $\Phi^{(2)}$ respectively, and $\lambda \in [0, 1]$. Then the portfolio given by the weighted average
\[
\pi := \lambda \pi^{(1)} + (1 - \lambda) \pi^{(2)}
\]
belongs to ${\mathcal{FG}}$. Indeed, $\pi$ is generated by the geometric mean 
\[
\Phi := \left( \Phi^{(1)} \right)^{\lambda}  \left( \Phi^{(2)} \right)^{1 - \lambda}
\]
of the two generating functions.
\end{lemma}
\begin{proof}
For $C^2$ generating functions this result is stated in \cite[Page 50]{F02}. The same is true in the general case where the generating functions are not necessarily smooth. To prove this, we need to check that $\pi = \lambda \pi^{(1)} + (1 - \lambda) \pi^{(2)}$ satisfies the defining inequality \eqref{eqn:superdiff}. This is an easy consequence of the AM-GM inequality and the proof is omitted. \qed
\end{proof}

\begin{lemma} \label{lem:Driftconcave}
The L-divergence and the drift quadratic form are concave in the portfolio weights in the following sense. Let $(\pi^{(1)}, \Phi^{(1)}), (\pi^{(2)}, \Phi^{(2)}) \in {\mathcal{FG}}$. For $\lambda \in [0, 1]$, let $\pi = \lambda \pi^{(1)} + (1 - \lambda) \pi^{(2)}$ and let $\Phi = \left(\Phi^{(1)}\right)^{\lambda} \left( \Phi^{(2)} \right)^{1 - \lambda}$ be the generating function of $\pi$. Let $T$, $T^{(1)}$ and $T^{(2)}$ be the L-divergences of $(\pi, \Phi)$, $(\pi^{(1)}, \Phi^{(1)})$ and $(\pi^{(2)}, \Phi^{(2)})$ respectively. Then
\begin{equation} \label{eqn:concavediscrete}
T\left(q \mid p \right) \geq \lambda T^{(1)}\left(q \mid p \right) + (1 - \lambda) T^{(2)}\left(q \mid p \right), \quad p, q \in \Delta^{(n)}.
\end{equation}
If $\Phi^{(1)}$ and $\Phi^{(2)}$ are $C^2$, then $H_{\pi} \geq \lambda H_{\pi^{(1)}} + (1 - \lambda) H_{\pi^{(2)}}$ in the sense that
\begin{equation} \label{eqn:concavecont}
H_{\pi}(p)(v, v) \geq \lambda H_{\pi^{(1)}}(p)(v, v) + (1 - \lambda) H_{\pi^{(2)}}(p)(v, v)
\end{equation}
for all $p \in \Delta^{(n)}$ and $v \in T\Delta^{(n)}$.
\end{lemma}
\begin{proof}
To prove \eqref{eqn:concavediscrete} we write the L-divergence $T\left(q \mid p \right)$ of a functionally generated portfolio $(\pi, \Phi)$ in the form
\[
T\left(q \mid p \right) = \log \left(1 + \left\langle \frac{\pi(p)}{p}, q - p \right\rangle\right) - I_{\pi}(\gamma),
\]
where $I_{\pi}(\gamma) = \int_{\gamma} \frac{\pi}{p} dp$ is the line integral of the weight ratio along the line segment from $p$ to $q$ (see \eqref{eqn:lineintegral}). Since the line integral is linear in $\pi$ and the logarithm is concave, we see that $T\left(q \mid p \right)$ is concave in $\pi$. The statement for the drift quadratic form follows from the Taylor approximation \eqref{eqn:taylor}.
\qed\end{proof}

We are now ready to prove Theorem \ref{thm:main}. 

\begin{proof}[Proof of Theorem \ref{thm:main}]
Let $\tau: \Delta^{(n)} \rightarrow \overline{\Delta^{(n)}}$ be a $C^1$ portfolio which dominates $\pi$ on compacts. We want to prove that $\tau = \pi$. By Theorem \ref{prop:MCM}, $\tau$ is generated by a concave function $\Psi: \Delta^{(n)} \rightarrow (0, \infty)$. Since $\tau$ is $C^1$, by \cite[Proposition 5(iii)]{PW14} $\Psi$ is $C^2$, so $\tau \in {\mathcal{FG}}^2$. Thus we may rephrase Theorem \ref{thm:main} by saying that $\pi$ is maximal in ${\mathcal{FG}}^2$.

Let $\Psi$ be a generating function of $\tau$. By scaling, we may assume that $\Psi(\overline{e}) = \Phi(\overline{e})$. We will prove that $\Psi$ equals $\Phi$ identically, so $\Psi$ generates $\pi$ and $\tau = \pi$. We divide the proof into the following steps.

\bigskip

\noindent
{\it Step 1 (Symmetrization).} Let $S_n$ be the set of permutations of $\{1, ..., n\}$. For $\sigma \in S_n$, define $\Psi_{\sigma}$ by relabelling the coordinates, i.e.,
\[
\Psi_{\sigma}(p) = \Psi(p_{\sigma(1)}, ..., p_{\sigma(n)}).
\]
Since $\tau \succeq \pi$, by Lemma \ref{lem:drift} (and relabeling the coordinates) we have $H_{\Psi_{\sigma}} \geq H_{\Phi_{\sigma}}$ for all $\sigma \in S_n$. But $\Phi$ is a measure of diversity, so $\Phi_{\sigma} = \Phi$ by symmetry and we have $H_{\Psi_{\sigma}} \geq H_{\Phi}$ for all $\sigma \in S_n$. Let
\[
\widetilde{\Psi} = \prod_{\sigma \in S_n} \left(\Psi_{\sigma}\right)^{\frac{1}{n!}}
\]
be the symmetrization of $\Psi$. By Lemma \ref{lem:FGconvex}, $\widetilde{\Psi}$ generates the symmetrized portfolio
\[
\widetilde{\tau}(p) = \frac{1}{n!} \sum_{\sigma \in S_n} \tau(p_{\sigma(1)}, ..., p_{\sigma(n)}), \quad p \in \Delta^{(n)}.
\]
By Lemma \ref{lem:Driftconcave}, we have
\begin{equation} \label{eqn:symmetrizedineq}
H_{\widetilde{\Psi}} \geq \frac{1}{n!} \sum_{\sigma \in S_n} H_{\Psi_{\sigma}} \geq H_{\Phi}.
\end{equation}
Thus $H_{\widetilde{\Psi}} \succeq H_{\Phi}$. Clearly $\widetilde{\Psi}$ is a measure of diversity and by symmetry it achieves its maximum at $\overline{e}$.

\bigskip

\noindent
{\it Step 2 ($\widetilde{\Psi} \leq \Phi$).} We claim that $\widetilde{\Psi} \leq \Phi$ on $\Delta^{(n)}$. Let $p \in \Delta^{(n)}$ and consider the one-dimensional concave functions
\begin{equation} \label{eqn:uandv}
\begin{split}
u(t) &= \Phi((1 - t)\overline{e} + tp) \\
v(t) &= \widetilde{\Psi}((1 - t)\overline{e} + tp)
\end{split}
\end{equation}
defined on $[0, 1]$. We have $u(0) = v(0)$ and $u'(0) = v'(0) = 0$ since both $\Phi$ and $\widetilde{\Psi}$ achieve their maximums at $\overline{e}$. Since $H_{\widetilde{\Psi}} \geq H_{\Phi}$, we have
\[
\frac{-v''(t)}{v(t)} \geq \frac{-u''(t)}{u(t)}, \quad t \in [0, 1].
\]
By the relative concavity lemma (Lemma \ref{lem:relativeconcavity}),
\begin{equation} \label{eqn:w}
w(y) = \frac{v(G(y))}{u(G(y))}
\end{equation}
is a positive concave function on $[0, \ell]$, where $\ell = \int_0^1 \frac{1}{u(t)^2} \mathrm{d}t$, with $w(0) = 1$ and $w'(0) = 0$ (by the quotient rule). Note that $\ell < \infty$ as $\Phi$ is continuous and positive on the line segment $[\overline{e}, p] \subset \Delta^{(n)}$. Also, it is straightforward to see that in this case the relative concavity lemma can be applied to $[0, \ell]$ instead of $[0, \ell)$. This implies that $w$ is non-increasing and so $w(\ell) = \widetilde{\Psi}(p) / \Phi(p) \leq 1$.

\bigskip

\noindent
{\it Step 3  ($\widetilde{\Psi} \equiv \Phi$).} Let $Z = \{p \in \Delta^{(n)}: \widetilde{\Psi}(p) = \Phi(p)\}$ and we claim that $Z = \Delta^{(n)}$. Here we follow an idea in the proof of \cite[Theorem 3]{CDOS09}. Define $u$ and $v$ on $[0, 1)$ by \eqref{eqn:uandv} with $p$ replaced by $e(1)$. Then the function $w$ defined as in \eqref{eqn:w} is positive and concave on $[0, \infty)$ since the integral in \eqref{eqn:integralcondition} (which defines $\ell = \int_0^1 \frac{1}{u(t)^2} \mathrm{d}t$) diverges. Again $w$ satisfies $w(0) = 1$ and $w'(0) = 0$. But since $w$ is defined on an infinite interval, if $w'(y) < 0$ for some $y$, then $w$ must hit zero as $w'$ is non-increasing by concavity. This contradicts the positivity of $w$, and so $w$ is identically one on $[0, \infty)$. It follows that $\widetilde{\Psi} = \Phi$ on the line segment $[\overline{e}, e(1))$. By symmetry, $Z$ contains the segments $[\overline{e}, e(i))$ for all $i$.

\medskip

Next we show that the set $Z$ is convex. Let $p, q \in Z$. Again we consider the pair of functions
\begin{equation} \label{eqn:uandvpq}
\begin{split}
u(t) &= \Phi((1 - t)p + tq) \\
v(t) &= \widetilde{\Psi}((1 - t)p + tq)
\end{split}
\end{equation}
on $[0, 1]$. Let $\widetilde{w}(t) = \frac{v(t)}{u(t)}$, $t \in [0, 1]$. By the relative concavity lemma again, we know that $\widetilde{w}$ is concave after a reparameterization. But $\widetilde{w}(t) \leq 1$ by Step 2 and $\widetilde{w}$ equals one at the endpoints $0$ and $1$. By concavity, $\widetilde{w}$ is identically one on $[0, 1]$. Hence if $Z$ contains $p$ and $q$, it also contains the line segment $[p, q]$. Now $Z$ is a convex set containing $[\overline{e}, e(i))$ for all $i$. It is easy to see that $Z$ is then the simplex $\Delta^{(n)}$. Hence $\widetilde{\Psi}$ equals $\Phi$ identically.

\bigskip

\noindent
{\it Step 4 (Desymmetrization).} We have shown that  $\widetilde{\Psi} \equiv \Phi$, and so $H_{\widetilde{\Psi}} = H_{\Phi}$. By \eqref{eqn:symmetrizedineq}, we have
\[
H_{\Phi} = H_{\widetilde{\Psi}} \geq \frac{1}{n!} \sum_{\sigma \in S_n} H_{\Psi_{\sigma}} \geq H_{\Phi}.
\]
Since $H_{\Psi_{\sigma}} \geq H_{\Phi}$ for each $\sigma \in S_n$, we have $H_{\Psi_{\sigma}} = H_{\Phi}$ for all $\sigma$. In particular, taking $\sigma$ to be the identity, we have $H_{\Psi} = H_{\Phi}$. It remains to show that $\Psi$ equals $\Phi$ identically (recall that we assume $\Psi(\overline{e}) = \Phi(\overline{e})$).

Fix $i \in \{1, ..., n\}$ and consider
\begin{equation*}
\begin{split}
u(t) &= \Phi((1 - t)\overline{e} + te(i)) \\
v(t) &= \Psi((1 - t)\overline{e} + te(i))
\end{split}
\end{equation*}
for $t \in [0, 1)$. By the argument in Step 3, if $\left(\frac{v}{u}\right)'(0) \leq 0$, the integral condition \eqref{eqn:integralcondition} implies that $v / u$ is identically one. So $\left(\frac{v}{u}\right)'(0) \leq 0$ implies $\left(\frac{v}{u}\right)'(0) = 0$. For $\sigma \in S_n$ let
\[
v_{\sigma}(t) = \Psi((1 - t)\overline{e} + te(\sigma(i))).
\]
Since $\widetilde{\Psi} = \Phi$, we have
\[
\prod_{\sigma \in S_n} \left(\frac{v_{\sigma}(t)}{u(t)}\right)^{\frac{1}{n!}} = 1.
\]
Taking logarithm on both sides and differentiating, we see that the average of the derivatives $\left(\frac{v}{u}\right)'(0)$ over $i$ is $0$ (recall that $\Phi$ is symmetric). Since all derivatives are non-negative by the above argument, in fact they are all $0$, and so $\Psi = \Phi$ on $[\overline{e}, e(i))$ for all $i$.

Since the vectors $e(i) - \overline{e}$ span the plane parallel to $\Delta^{(n)}$, the graphs of $\Psi$ and $\Phi$ have the same tangent plane at $\overline{e}$. Since $\Phi$ achieves its maximum at $\overline{e}$, we see that $\Psi$ achieves its maximum at $\overline{e}$ as well. Now we may apply the argument in Steps 2 and 3 to conclude that $\Psi$ equals $\Phi$ identically on $\Delta^{(n)}$. Thus $\tau = \pi$ and we have proved that $\pi$ is maximal in ${\mathcal{FG}}^2$.
\qed\end{proof}

\begin{proof}[Proof of Corollary \ref{cor:main}]
Let $\tau$ be a $C^1$ portfolio not equal to $\pi$. By the maximality of $\pi$, it is not the case that $\tau \succeq \pi$. By Theorem \ref{prop:MCM}, $\tau$ does not satisfy MCM relative to $\pi$. Thus, there is a cycle $\{\mu(t)\}_{t = 0}^{m+1}$ (with $\mu(0) = \mu(m + 1)$) over which
\begin{equation} \label{eqn:badcycle}
\frac{V_{\tau}(m+1)}{V_{\pi}(m+1)} < 1.
\end{equation}
Consider, as in the proof of Theorem \ref{prop:MCM}, the market weight sequence which goes through this cycle again and again. Clearly $\{\mu(t)\}_{t \geq 0}$ takes values in a finite set $K$ which is compact. From \eqref{eqn:badcycle}, it is clear that $V_{\tau}(t) / V_{\pi}(t) \rightarrow 0$ as $t \rightarrow \infty$.
\qed\end{proof}

\subsection{Extension to continuous time} \label{sec:continuoustime}
We discuss briefly how Theorem \ref{thm:main} can be generalized to continuous time. In continuous time, we let the market weight process $\{\mu(t)\}_{t \geq 0}$ be a continuous semimartingale with state space $\Delta^{(n)}$. The market weight process of a portfolio $\pi$ satisfies the stochastic differential equation
\[
\frac{\mathrm{d}V_{\pi}(t)}{V_{\pi}(t)} = \sum_{i = 1}^n \pi_i(\mu(t)) \frac{\mathrm{d}\mu_i(t)}{\mu_i(t)}.
\]
Let $(\pi, \Phi), (\tau, \Psi) \in {\mathcal{FG}}^2$. Then we have the decomposition
\[
\log \frac{V_{\tau}(t)}{V_{\pi}(t)} = \log \frac{\Psi(\mu(t)) / \Psi(\mu(0))}{\Phi(\mu(t)) / \Phi(\mu(0))} + A(t),
\]
where the drift process takes the form $A(t) = A_{\tau}(t) - A_{\pi}(t)$,
\begin{equation} \label{eqn:qvariation}
A_{\tau}(t) = \int_0^t H_{\tau}(\mu(s))(\mathrm{d}\mu(s), \mathrm{d}\mu(s)),
\end{equation}
and the analogous definition holds for $A_{\pi}$. See \cite[Theorem 3.1.5]{F02}. In \eqref{eqn:qvariation} we use the intrinsic notation of \cite{EM89} for the quadratic variation of $\{\mu(t)\}$ with respect to the non-negative definite form $H_{\tau}$. It can be shown that $A(t)$ is non-decreasing almost surely for all continuous semimartingales $\{\mu(t)\}$ if and only if $A_{\tau} \geq A_{\pi}$. We may define the relation $\tau \succeq \pi$ (domination on compacts) in the same way as in Definition \ref{def:pseudoarbitrage}, except that we require for any continuous semimartingale $\{\mu(t)\}$ with values in $K$, \eqref{eqn:lowerbound} holds for all $t \geq 0$ almost surely. Using the results established, one can show in continuous time that $\pi$ is maximal in ${\mathcal{FG}}^2$ if it is generated by a measure of diversity satisfying \eqref{eqn:integralcondition}.

Moreover, in continuous time, \cite[Theorem 3]{CDOS09} shows that the integral condition \eqref{eqn:integralcondition} is also necessary for $(\pi, \Phi)$ to be maximal in ${\mathcal{FG}}^2$ when $n = 2$. Let $u(x) = \Phi(x, 1 - x)$. The idea is that if the integral converges, we can solve the initial value problem
\[
v''(x) + \left(\frac{-u''(x)}{u(x)} + s(x)\right)v(x) = 0, \quad x \in (0, 1),
\]
\[
v\left(\frac{1}{2}\right) = u\left(\frac{1}{2}\right), \quad v'\left(\frac{1}{2}\right) = u'\left(\frac{1}{2}\right) = 0,
\]
for some appropriately chosen function $s(x)$ such that $s(x) \geq 0$, $s(x) \not\equiv 0$ and $s$ is symmetric about $\frac{1}{2}$. Sturm's comparison theorem implies that the solution $v(x)$ is positive (and concave) on $(0, 1)$. Let $\Psi(p) = v(p_1)$ and let $\tau$ be the portfolio generated by $\Psi$. Then the corresponding portfolio $\tau$ is not equal to $\pi$ and dominates $\pi$ on compacts, so $\pi$ is not maximal in ${\mathcal{FG}}^2$.

\begin{problem}
Characterize the maximal portfolios of ${\mathcal{FG}}$.
\end{problem}

\section{Optimization of functionally generated portfolios} \label{sec:optimization}
\subsection{A shape-constrained optimization problem} \label{sec:optim2}
Consider the relative value process of a functionally generated portfolio. If we have a model for the market weight process $\{\mu(t)\}_{t \geq 0}$, a natural optimization problem is to maximize the expected growth rate of the drift  process over some horizon. To this end, suppose we are given an {\it intensity measure} ${\Bbb P}$ of the increments $(\mu(t), \mu(t + 1))$ modeled as a Borel probability measure on $\Delta^{(n)} \times \Delta^{(n)}$. We assume that ${\Bbb P}$ is either discrete (with countably many masses) or absolutely continuous with respect to the measure $\nu := m \otimes m$ on $\Delta^{(n)} \times \Delta^{(n)}$, where $m$ is the surface measure of $\Delta^{(n)}$ in ${\Bbb R}^n$ (which should be thought of as the Lebesgue measure on $\Delta^{(n)}$). We will abbreviate this by simply saying ${\Bbb P}$ is absolutely continuous. For technical reasons, we assume that ${\Bbb P}$ is supported on $K \times K$ for some compact subset $K$ of $\Delta^{(n)} \times \Delta^{(n)}$.

Given the intensity measure ${\Bbb P}$, we consider the optimization problem
\begin{equation} \label{eqn:newoptim1}
\max_{(\pi, \Phi) \in {\mathcal{FG}}} \int T\left(q \mid p\right) \mathrm{d} {\Bbb P}.
\end{equation}

First we give some examples of the intensity measure.

\begin{example} \label{exm:markov}
Suppose $\{(\mu(t - 1), \mu(t))\}$ is an ergodic Markov chain on $K \times K$. We can take ${\Bbb P}$ to be the stationary distribution of $(\mu(t - 1), \mu(t))$. It is easy to see that an optimal portfolio in \eqref{eqn:newoptim1} maximizes the asymptotic growth rate $\lim_{t \rightarrow \infty} \frac{1}{t} \log V_{\pi}(t)$ of the relative value (the term $\frac{1}{t} \log \frac{\Phi(\mu(t))}{\Phi(\mu(0))}$ vanishes as $t \rightarrow \infty$). This portfolio can be regarded as a {\it growth optimal portfolio} (relative to the market portfolio) among the functionally generated portfolios.
\end{example}

\begin{example} \label{exm:green}
We model $\{\mu(t)\}_{t \geq 0}$ as a stochastic process. Let $K$ be a compact subset of $\Delta^{(n)}$ containing $\mu(0)$. Let $\tau$ be the first exit time of $K$, i.e.,
\[
\tau = \inf\{t \geq 0: \mu(t) \notin K\}.
\]
Consider the measure ${\Bbb G}$ on $K \times K$ defined by
\[
{\Bbb G}(A) := {\Bbb E} \left[\sum_{t = 1}^{\tau - 1} 1_{\{(\mu(t-1), \mu(t)) \in A\}} \right], \quad A \subset K \times K \text{ measurable}.
\]
If the process $\{(\mu(t-1), \mu(t))\}$ is Markovian, ${\Bbb G}$ is the {\it Green kernel} of the process killed at time $\tau$. Suppose ${\Bbb G}(K \times K) = {\Bbb E} (\tau - 1) < \infty$, i.e., the exit time has finite expectation. Then
\begin{equation*}
{\Bbb P}(\cdot) := \frac{1}{{\Bbb G}(K \times K)} {\Bbb G}(\cdot)
\end{equation*}
is a probability measure on $K \times K$. This intensity measure will be used in the empirical example in Section \ref{sec:empirical}.
\end{example}

Note that Example \ref{exm:markov} deals with infinite horizon while Example \ref{exm:green} is concerned with a finite (but random) horizon. The optimization problem \eqref{eqn:newoptim1} is shape-constrained because the generating function is concave by definition. We will first study some theoretical properties of this abstract (unconstrained) optimization problem, and then focus on a discrete special case where numerical solutions are possible and further constraints are imposed. In contrast to classical portfolio selection theory where the portfolio weights are optimized period by period, in \eqref{eqn:newoptim1} we optimize the portfolio weights over a region simultaneously.

Throughout the development it is helpful to keep in mind the analogy between \eqref{eqn:newoptim1} and the maximum likelihood estimation of a log-concave density. In that context, we are given a random sample $X_1, ..., X_N$ from a log-concave density $f_0$ on ${\Bbb R}^d$ (i.e., $\log f_0$ is concave). The log-concave maximum likelihood estimate (MLE) $\widehat{f}$ is the solution to
\begin{equation} \label{eqn:MLE}
\max_f \sum_{j = 1}^N \log f(X_j),
\end{equation}
where $f$ ranges over all log-concave densities on ${\Bbb R}^d$. It can be shown that the MLE exists almost surely (when $N \geq d + 1$ and the support of $f_0$ has full dimension) and is unique; see \cite{CSS10} for precise statements of these results. We remark that \eqref{eqn:newoptim1} is more complicated than \eqref{eqn:MLE} because the portfolio weights correspond to selections of the superdifferential $\partial \log \Phi$, wheras \eqref{eqn:MLE} involves only the values of the density.

\subsection{Theoretical properties}
It is easy to check that \eqref{eqn:newoptim1} is a convex optimization problem since the L-divergence is concave in the portfolio weights (Lemma \ref{lem:Driftconcave}). First we show that \eqref{eqn:newoptim1} has an optimal solution and study in what sense the solution is unique.

\medskip

Given an intensity measure ${\Bbb P}$, it can be decomposed in the form
\begin{equation} \label{eqn:conditional}
{\Bbb P}(\mathrm{d}p\mathrm{d}q) = {\Bbb P}_1(\mathrm{d}p) {\Bbb P}_2(\mathrm{d}q | p),
\end{equation}
where ${\Bbb P}_1$ is the first marginal of ${\Bbb P}$ and ${\Bbb P}_2$ is the conditional distribution of the second variable given $p$. We will need a technical condition for ${\Bbb P}$ which allows jumps in all directions. 

\begin{definition}[Support condition]
Let ${\Bbb P}$ be an absolutely continuous probability measure on $\Delta^{(n)} \times \Delta^{(n)}$ with the decomposition \eqref{eqn:conditional}. Write
\[
{\Bbb P}_1(\mathrm{d}p) = f(p) m(\mathrm{d}p),
\]
where $f(\cdot)$ is the density of ${\Bbb P}_1$ with respect to $m$. We say that ${\Bbb P}$ satisfies the support condition if for $m$-almost all $p$ for which $f(p) > 0$, for all $v \in T\Delta^{(n)}$, there exists $\lambda > 0$ such that $p + \lambda v$ belongs to the support of ${\Bbb P}_2(\cdot | p)$.
\end{definition}

We have the following result which is analogous to \cite[Theorem 1]{CSS10}.

\begin{theorem} \label{thm:optim}
Consider the optimization problem \eqref{eqn:newoptim1} where ${\Bbb P}$ is a discrete or absolutely continuous Borel probability measure on $\Delta^{(n)} \times \Delta^{(n)}$ supported on $K \times K$ with $K \subset \Delta^{(n)}$ compact.
\begin{enumerate}
\item[(i)] The problem has an optimal solution. 
\item[(ii)] If $\pi^{(1)}$ and $\pi^{(2)}$ are optimal solutions, then
\begin{equation} \label{eqn:uniqueness}
\left\langle \frac{\pi^{(1)}(p)}{p}, q - p \right\rangle = \left\langle\frac{\pi^{(2)}(p)}{p}, q - p \right\rangle
\end{equation}
for ${\Bbb P}$-almost all $(p, q)$. In particular, if ${\Bbb P}(\mathrm{d}p\mathrm{d}q) = {\Bbb P}_1(\mathrm{d}p) {\Bbb P}_2(\mathrm{d}q | p)$ is absolutely continuous with ${\Bbb P}_1(\mathrm{d}p) = f(p) m(\mathrm{d}p)$ and satisfies the support condition, then $\pi^{(1)} = \pi^{(2)}$ $m$-almost everywhere on $\{p: f(p) > 0\}$.
\end{enumerate}
\end{theorem}

The proofs of Theorem \ref{thm:optim} and Theorem \ref{thm:optim2} below are given in Appendix \ref{sec:appendix}.

\medskip

Let ${\Bbb P}$ an intensity measure. Suppose $\{{\Bbb P}_N\}_{N \geq 1}$ is a sequence of probability measures converging weakly to ${\Bbb P}$. By definition, this means that
\begin{equation*} 
\lim_{N \rightarrow \infty} \int f \mathrm{d}{\Bbb P}_N = \int f \mathrm{d}{\Bbb P}
\end{equation*}
for all bounded continuous functions on $\Delta^{(n)} \times \Delta^{(n)}$. For example, one may sample i.i.d.~observations $\{(p(j), q(j))\}_{j = 1}^N$ from ${\Bbb P}$ and take ${\Bbb P}_N$ to be the empirical measure $\frac{1}{N} \sum_{j = 1}^N \delta_{(p(j), q(j))}$, where $\delta_{(p(j), q(j))}$ is the point mass at $(p(j), q(j))$. From the perspective of statistical inference, the optimal portfolio $(\widehat{\pi}^{(N)}, \widehat{\Phi}^{(N)})$ for ${\Bbb P}_N$ can be regarded as a point estimate of the optimal portfolio $(\pi, \Phi)$ for ${\Bbb P}$. The following result states that the estimator is consistent. See \cite[Theorem 4]{CS10} for an analogous statement in the context of log-concave density estimation.

\begin{theorem} \label{thm:optim2}
Let $(\pi, \Phi)$ be the optimal portfolio in problem \eqref{eqn:newoptim1} for ${\Bbb P}$, where ${\Bbb P}(\mathrm{d}p\mathrm{d}q) = {\Bbb P}_1(\mathrm{d}p) {\Bbb P}_2(\mathrm{d}q | p)$ is absolutely continuous with ${\Bbb P}_1(\mathrm{d}p) = f(p) m(\mathrm{d}p)$, supported on $K \times K$ with $K \subset \Delta^{(n)}$ compact, and satisfies the support condition. Let $\{{\Bbb P}_N\}$ be a sequence of discrete or absolutely continuous probability measures on $K \times K$ such that ${\Bbb P}_N \rightarrow {\Bbb P}$ weakly, and suppose $(\widehat{\pi}^{(N)}, \widehat{\Phi}^{(N)})$ is optimal for the measure ${\Bbb P}_N$, $N \geq 1$. Then $\widehat{\pi}^{(N)} \rightarrow \pi$ $m$-almost everywhere on $\{p: f(p) > 0\}$.
\end{theorem}

\subsection{Finite dimensional reduction}
Without further constraints, the optimal portfolio weights of \eqref{eqn:newoptim1} may be highly irregular. Now we restrict to the special case where
\begin{equation} \label{eqn:discreteP}
{\Bbb P} = \frac{1}{N} \sum_{j = 1}^N \delta_{(p(j), q(j))}
\end{equation}
is a discrete measure and $(p(j), q(j)) \in \Delta^{(n)} \times \Delta^{(n)}$ for $j = 1, ..., N$. This presents no great loss of generality because in practice the market weights have finite precision and we can choose the pairs $(p(j), q(j))$ to take values in a grid approximating $\Delta^{(n)} \times \Delta^{(n)}$. Moreover, from Theorem \ref{thm:optim2} we expect that when $N$ is large the optimal solution approximates that of the continuous counterpart. Consider the modified optimization problem
\begin{equation} \label{eqn:newoptim2}
\begin{aligned}
& \underset{(\pi, \Phi) \in {\mathcal{FG}}}{\text{maximize}}
& &  \int T\left(q \mid p\right)\mathrm{d} {\Bbb P} \\
& \text{subject to}
& & (\pi(p(1)), ..., \pi(p(N))) \in C, \\
\end{aligned}
\end{equation}
where $C$ is a given closed convex subset of $\overline{\Delta^{(n)}}^N$. Some examples of $C$ are given in Table \ref{tab:constraints}, where each constraint is a cylinder set of the form $\{\pi(p(j)) \in C_j\}$ with $C_j$ a closed convex set of $\overline{\Delta^{(n)}}$. `Global' constraints on the weights can be imposed, see Section \ref{sec:empirical} for an example. It can be verified easily that the proof of Theorem \ref{thm:optim} goes through without changes with these constraints, so \eqref{eqn:newoptim2} has an optimal solution. Moreover, if $\pi^{(1)}$ and $\pi^{(2)}$ are optimal solutions, then
\[
\left\langle \frac{\pi^{(1)}(p(j))}{p(j)}, q(j) - p(j) \right\rangle = \left\langle\frac{\pi^{(2)}(p(j))}{p(j)}, q(j) - p(j) \right\rangle, \quad j = 1, ..., N.
\]

\begin{table}
\caption{Examples of additional constraints imposed for $p \in \{p(1), ..., p(N)\}$. The parameters may be given functions of $p$.}
 \label{tab:constraints}
    \begin{tabular}{ll}
\hline\noalign{\smallskip}
   Constraint  & Interpretation \\ 
\noalign{\smallskip}\hline\noalign{\smallskip}
    $a_i \leq \pi_i(p) \leq b_i$ & Box constraints on portfolio weights  \\ 
    $m_i \leq \frac{\pi_i(p)}{p_i} \leq M_i$ & Box constraints on weight ratios \\
     $(\pi(p) - p)' \Sigma (\pi(p) - p) < \varepsilon$ & Constraint on tracking error given a covariance matrix   \\ \noalign{\smallskip}\hline
\end{tabular}
\end{table}

For maximum likehood estimation of log-concave density, it is shown in \cite{CSS10} that the logarithm of the MLE $\widehat{f}$ is {\it polyhedral}, i.e., $\log \widehat{f}$ is the pointwise minimum of several affine functions (see \cite[Section 19]{R70}). In particular, there exists a triangulation of the data points over which $\log \widehat{f}$ is piecewise affine. We show that an analogous statement holds for \eqref{eqn:newoptim2}. Let $D = \{p(j), q(j): j = 1, ..., N\}$ be the set of data points.

\begin{theorem} \label{thm:optim3}
Let $(\pi, \Phi)$ be an optimal portfolio for the problem \eqref{eqn:newoptim2} where ${\Bbb P} = \frac{1}{N} \sum_{j = 1}^N \delta_{(p(j), q(j))}$. Let $\overline{\Phi}: \Delta^{(n)} \rightarrow (0, \infty)$ be the smallest positive concave function on $\Delta^{(n)}$ such that $\overline{\Phi}(p) \geq {\Phi}(p)$ for all $x \in D$. Then $\overline{\Phi}$ is a polyhedral positive concave function on $\Delta^{(n)}$ satisfying $\overline{\Phi} \leq \Phi$ and $\overline{\Phi}(p) = \Phi(p)$ for all $p \in D$. Moreover, $\overline{\Phi}$ generates a portfolio $\overline{\pi}$ such that $\overline{\pi}(p(j)) = \pi(p(j))$ for all $j$. In particular, $(\overline{\pi}, \overline{\Phi})$ is also optimal for the problem \eqref{eqn:newoptim2}.
\end{theorem}
\begin{proof}
It is a standard result in convex analysis that $\overline{\Phi}$ such defined is {\it finitely generated} (see \cite[Section 19]{R70}). By \cite[Corollary 19.1.2]{R70}, $\overline{\Phi}$ is a polyhedral concave function. By definition of $\overline{\Phi}$ and concavity of $\Phi$, we have $\overline{\Phi}(p) = {\Phi}(p)$ for all $x \in D$ for all $j$ and $\overline{\Phi} \leq \Phi$. This implies that $\partial \log \Phi(p(j)) \subset \partial \log \overline{\Phi}(p(j))$ for all $j$. By Lemma \ref{lem:superdiff}(ii), $\overline{\Phi}$ generates a portfolio $\overline{\pi}$ which agrees with $\pi$ on $\{p(1) ,..., p(N)\}$. It follows that (using obvious notations)
\[
\overline{T}\left(q(j) \mid p(j)\right) = T\left(q(j) \mid p(j) \right)
\]
for all $j$, and hence $(\overline{\pi}, \overline{\Phi})$ is optimal for \eqref{eqn:newoptim2}.
\qed\end{proof}

Theorem \ref{thm:optim3} reduces \eqref{eqn:newoptim2} to a finite-dimensional problem. In the next section we present an elementary implementation for the case $n = 2$ (analogous to univariate density estimation) and illustrate its application in portfolio management with a case study.

\section{Empirical examples} \label{sec:empirical}
\subsection{A case study}
\begin{figure}
\includegraphics[scale=0.5]{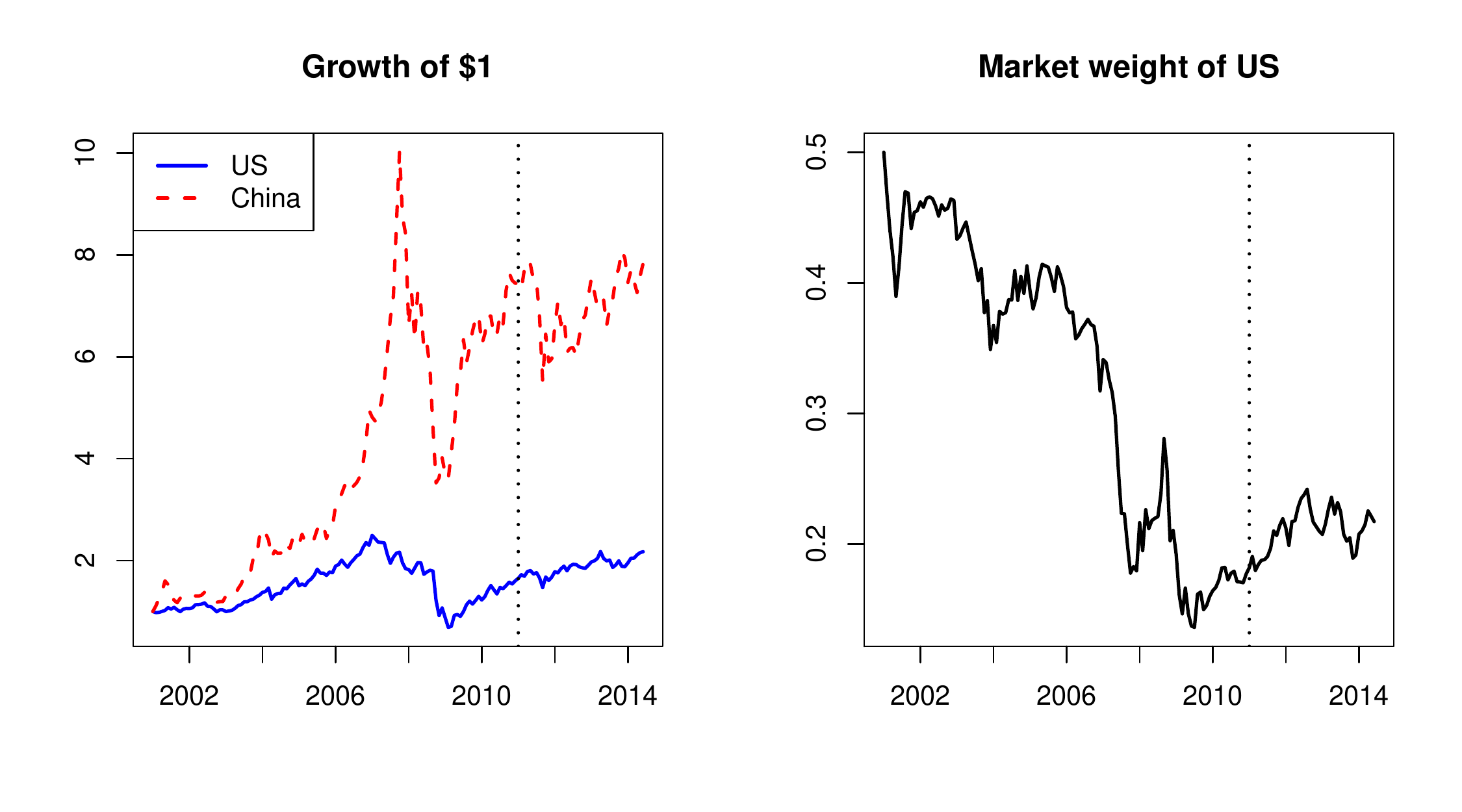}
\vspace{-20pt}
\caption{The figure on the left shows the growth of $\$1$ for each asset, and the one on the right shows the time series of the market weight $\mu_1(t)$ of US. The vertical dotted line divides the data set into the training and testing periods respectively.}
\label{fig:data}
\end{figure}

In global portfolio management, an important topic is the determination of the aggregate portfolio weights for countries. In this example we consider two countries: US and China. We represent them by the S\&P US BMI index (asset 1) and the S\&P China BMI index (asset 2) respectively. The `market' consists of these two assets. We collect monthly data from January 2001 to June 2014 using Bloomberg. The benchmark portfolio is taken to be the buy-and-hold portfolio starting with weights $(0.5, 0.5)$ at January 2001. Here the initial market weights $(0.5, 0.5)$ are chosen arbitrarily. The data from January 2001 to December 2010 will be used as the training data to optimize the portfolio which will be backtested in the subsequent period. The market weights at January 2011 are $(0.1819, 0.8191)$. The data is plotted in Figure \ref{fig:data}.

Let $K \subset \Delta^{(2)}$ be the compact set defined by
\begin{equation} \label{eqn:set}
K = \{p = (p_1, p_2) \in \Delta^{(2)}: 0.1 \leq p_1 \leq 0.3\}.
\end{equation}
Our objective here is to optimize a functionally generated portfolio to be held as long as the market weights stay within $K$. If the market weight of US approach these boundary points (regarded as a regime change), a new portfolio will be chosen, so $0.1$ and $0.3$ can be thought of as the {\it trigger points}.

\subsection{The intensity measure and constraints} \label{subsec:example}
\begin{figure}
\includegraphics[scale=0.4]{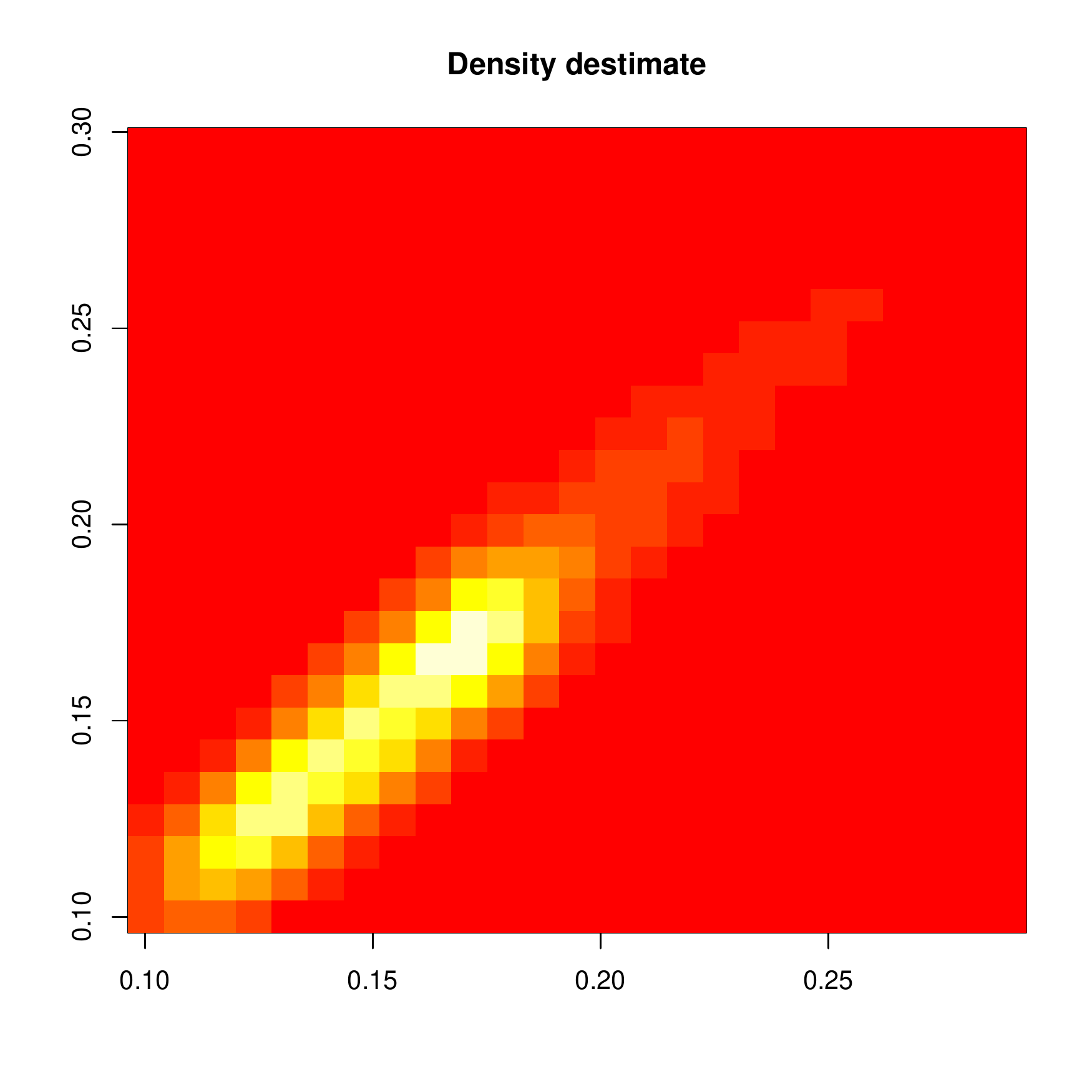}
\vspace{-15pt}
\caption{Density estimate of ${\Bbb P}_N$ on $K \times K$ in terms of the market weight of US.}
\label{fig:density}
\end{figure}

Suppose $t = 0$ corresponds to January 2011. We model $\{\mu(t)\}_{t \geq 0}$ as a discrete-time stochastic process (time is monthly) where $\mu(0)$ is constant. Let ${\Bbb P}$ be the measure in Example \ref{exm:green} where $\tau$ is the first exit time of $K$ given in \eqref{eqn:set}.

If a stochastic model is given, we may approximate ${\Bbb P}$ by simulating paths of $\{\mu(t)\}$ killed upon exiting $K$. The resulting empirical measure
\[
{\Bbb P}_N = \frac{1}{N} \sum_{j = 1}^N \delta_{(p(j), q(j))}
\]
is then taken as the intensity measure of the optimization problem \eqref{eqn:newoptim2}.

Since our main concern is the implementation of the optimization problem \eqref{eqn:newoptim2}, sophisticated modeling of $\{\mu(t)\}$ will not be attempted and we will use a simple method to simulate paths of $\{\mu(t)\}$. Namely, starting at $\mu(0) =  (0.1819, 0.8191)$, we simulate paths of $\{\mu(t)\}_{t = 0}^{\tau - 1}$ by {\it bootstrapping} the past returns of the two assets and computing the corresponding market weight series. In view of the possible recovery of US, before the simulation we recentered the past returns so that they both have mean zero over the training period. (Essentially, only the difference in returns matter for the evolution of the market weights.) We simulated 50 such paths and obtained $N = 3115$ pairs $(p(j), q(j))$ in $K \times K$. A density estimate of ${\Bbb P}_N$ (in terms of the market weight of US) is plotted in Figure \ref{fig:density}. To reduce the number of variables, the market weights are rounded to 3 decimal places, so the market weights of US take values in the set $D = \{0.100, 0.101, ..., 0.299, 0.300\}$.

Next we specify the constraints for $\{\pi(p_1) := \pi(p_1, 1 - p_1): p_1 \in D\}$. (This notation should cause no confusion since the market weight of China is determined by that of US.) First, we require that $\pi_1(p_1)$ is non-decreasing in $p_1$, i.e.,
\[
\pi_1(0.100) \leq \pi_1(0.101) \leq \cdots \leq \pi_1(0.300).
\]
This imposes a shape constraint on the portfolio weights which guarantees that the portfolio weights always move in the direction of market movement. To control the concentration of the portfolio we require also that the weight ratio of US satisfies $0.5 \leq \frac{\pi_1(p_1)}{p_1} \leq 2$ for $p_1 \in D$ (since there are only two assets, this implies a weight ratio bound for China). These constraints determine the convex set $C$ in the optimization problem \eqref{eqn:newoptim2} we are about to solve.

\subsection{Optimization procedure}
\begin{figure}
\includegraphics[scale=0.50]{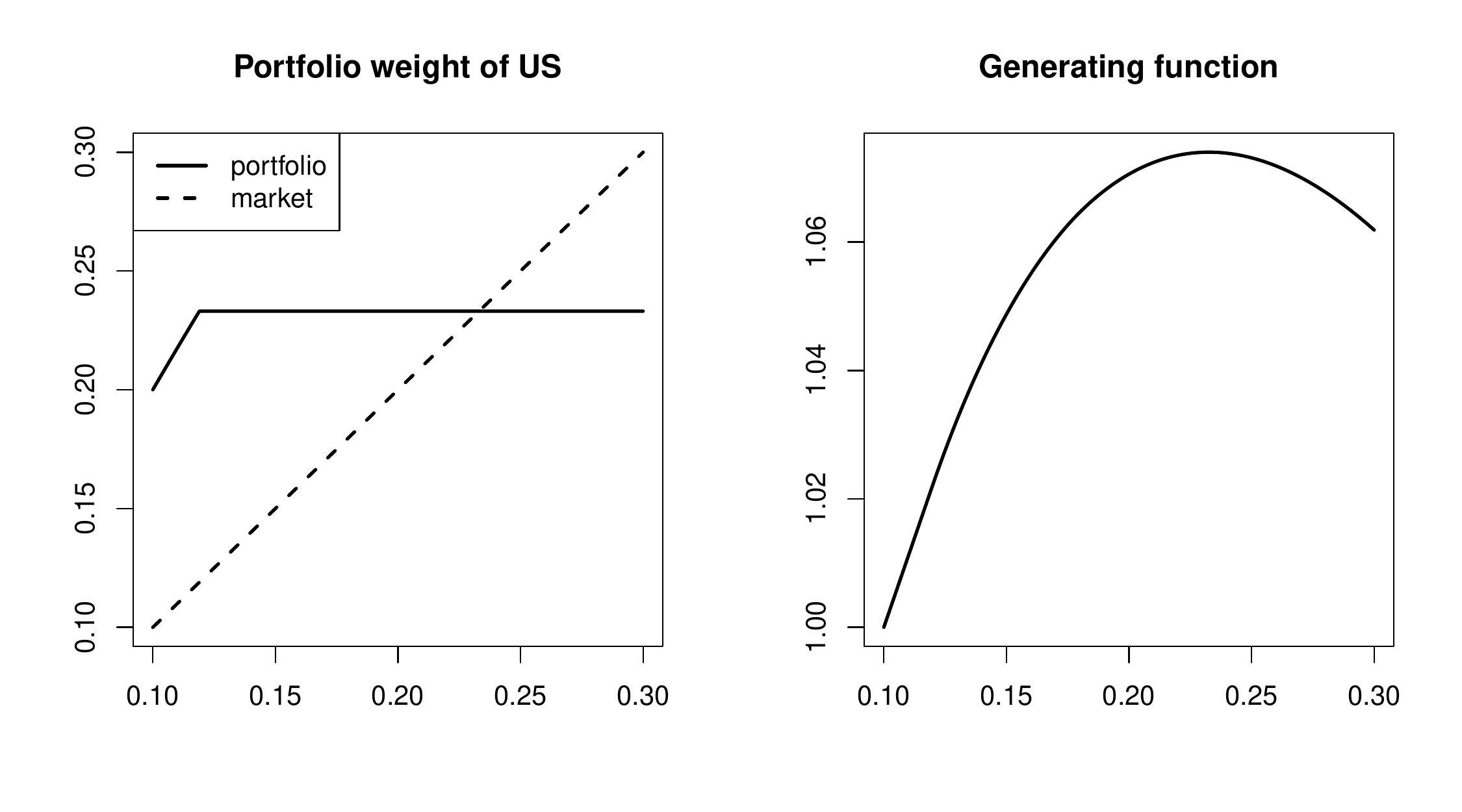}
\vspace{-20pt}
\caption{The portfolio weight and the generating function of the optimized portfolio.}
\label{fig:result}
\end{figure}

By Theorem \ref{thm:optim3}, it suffices to optimize over generating functions that are piecewise linear over the data points. First we introduce some simplifying notations. Write the set of grid points as $D = \{x_1 < x_2 < \cdots < x_m\}$ and let $x_0 = 0$, $x_{m+1} = 1$ be the endpoints of the interval. Let the decision variables be
\[
z_j := \pi(x_j, 1 - x_j), \quad j = 1, ..., m, 
\]
\[
\varphi_j := \Phi(x_j, 1 - x_j), \quad j = 0, ..., m + 1.
\]
By scaling, we may assume $\varphi_1 = 1$. The constraints on $\{\varphi_j\}$ are
\begin{equation} \label{eqn:nonnegative}
\varphi_j \geq 0, \quad j = 0, ..., m+1, \quad \varphi_1 = 1, \quad \text{(non-negativity)}
\end{equation}
\begin{equation} \label{eqn:concavity}
s_0 \geq s_1 \geq \cdots \geq s_m, \quad s_j := \frac{\varphi_{j+1} - \varphi_j}{x_{j+1} - x_j}. \quad \text{(concavity)}
\end{equation}
We require that $\pi$ is generated by $\Phi$. By \eqref{eqn:n=2weight} and Lemma \ref{lem:superdiff}, it can be seen that $z_j$ satisfies the inequality
\begin{equation} \label{eqn:generated}
x_j + x_j(1 - x_j) \frac{s_j}{\varphi_j} \leq z_j \leq x_j + x_j(1 - x_j) \frac{s_{j-1}}{\varphi_j}, \quad j = 1, ..., m. \quad \text{($(\pi, \Phi) \in {\mathcal{FG}}$)}
\end{equation}
We require that $z_j$ is non-decreasing in $j$:
\begin{equation}  \label{eqn:monotone}
z_1 \leq z_2 \leq \cdots \leq z_m. \quad \text{(monotonicity)}
\end{equation}
Finally, we require that the weight ratios are bounded between $0.5$ and $2$:
\begin{equation}  \label{eqn:weightratio}
0.5 \leq \frac{z_j}{x_j} \leq 2, \quad j = 1, ..., m.  \quad \text{(weight ratios)}
\end{equation}
With the constraints \eqref{eqn:nonnegative}-\eqref{eqn:weightratio} we maximize
\[
\int T\left(q \mid p\right) \mathrm{d}{\Bbb P}_N = \frac{1}{N} \sum_{j = 1}^N T\left( q(j) \mid p(j) \right)
\]
over $\{z_j\}$ and $\{\varphi_j\}$. This is a standard non-linear, but smooth, constrained optimization problem (convexity is lost because $\Phi$ is now piecewise linear). We implement this optimization problem using the \verb"fmincon" function in \verb"MATLAB". The optimal portfolio weights together with the generating function are plotted in Figure \ref{fig:result}. It turns out that the optimal portfolio is close to constant-weighted (with weights $(0.2331, 0.7669)$). Note that the constraint on the weight ratio limits the deviation of $\pi_1(p_1)$ from the market weight $p_1$. If the weight ratio constraint was not imposed (while the monotonicity constraint was kept), the optimal portfolio would be the equal-weighted portfolio $\pi \equiv (0.5, 0.5)$, and the reason can be seen from the proof of Lemma \ref{lem:diffeqn}.

\subsection{Backtesting the portfolio}
\begin{figure}
\includegraphics[scale=0.50]{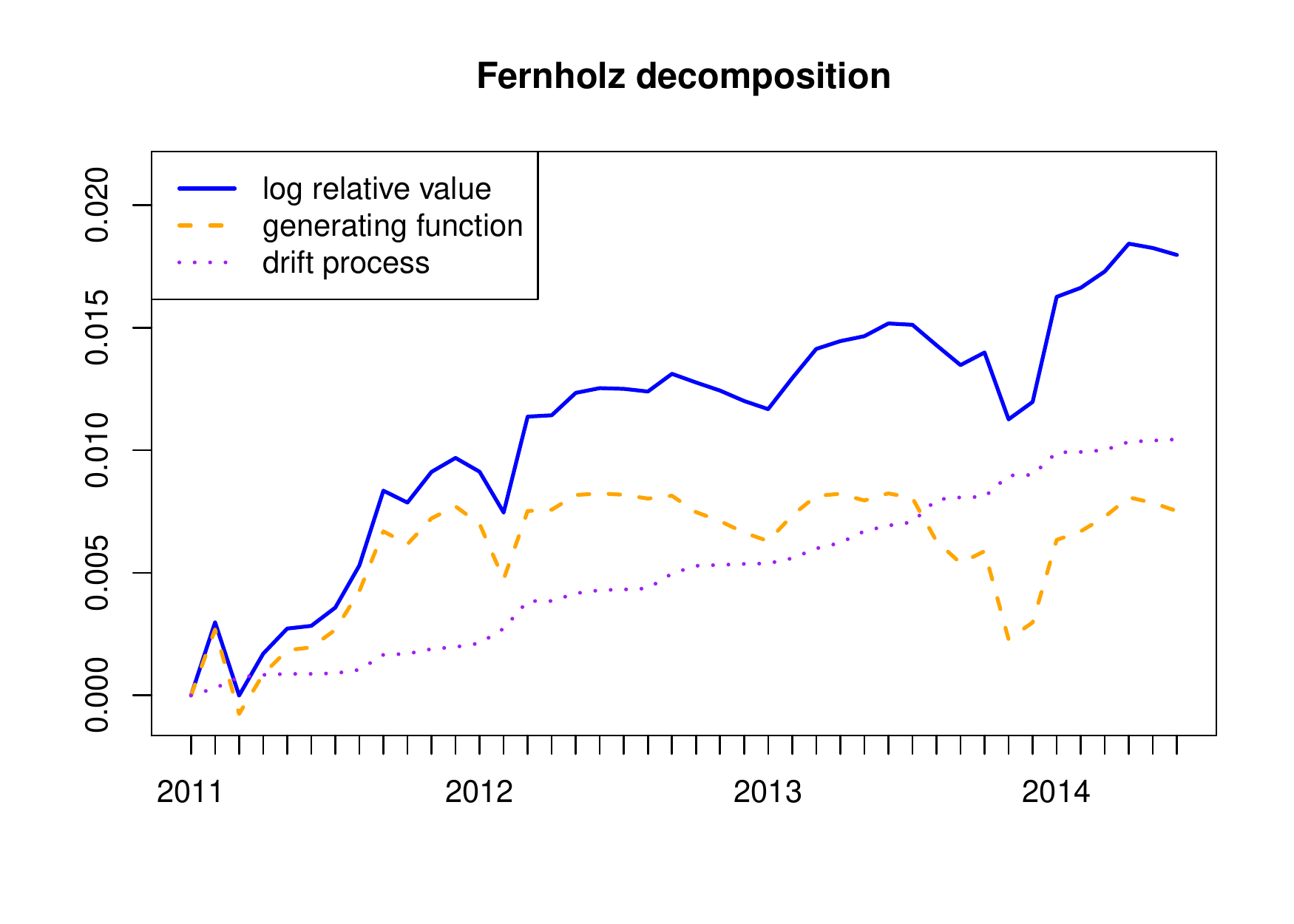}
\vspace{-20pt}
\caption{Fernholz's decomposition of the optimized portfolio over the testing period. The log relative value is $\log V_{\pi}(t)$. The generating function term is $\log \Phi(\mu(t)) - \log \Phi(\mu(0))$, and the drift process is $A(t)$.}
\label{fig:backtest}
\end{figure}

Finally, we compute the performance of the optimized portfolio over the testing period January 2011 to June 2014. The result (plotted using the function \verb"FernholzDecomp" of the \verb"RelValAnalysis" package) is shown in Figure \ref{fig:backtest}. Over the testing period, the portfolio beats the market by nearly 2\% in log scale and its performance has been steady. From the decomposition, about half of the outperformance is attributed to the increase of the generating function (note that the market weight of US becomes closer to $0.2331$ where the generating function attains its maximum), and the rest comes from the drift process. That the optimal portfolio is close to constant-weighted may not be very interesting, but this is a consequence of the data and our choice of constraints and is by no means obvious. Our optimization framework allows many other possibilities especially when there are multiple assets. Other useful constraints and efficient algorithms are natural subjects of further research.

\appendix
\section{Proofs of Theorem \ref{thm:optim} and Theorem \ref{thm:optim2}} \label{sec:appendix}
First we will state and prove some lemmas from convex analysis.

\begin{lemma} \label{lem:compactness}
Let $p_0 \in \Delta^{(n)}$ be fixed and let ${\mathcal C}_0$ be the collection of positive concave functions $\Phi$ on $\Delta^{(n)}$ satisfying $\Phi(p_0) = 1$. Then any sequence in ${\mathcal C}_0$ has a subsequence which converges locally uniformly on $\Delta^{(n)}$ to a function in ${\mathcal C}_0$.
\end{lemma}
\begin{proof}
By \cite[Theorem 10.9]{R70}, it suffices to prove that ${\mathcal C}_0$ has a uniform upper bound (the lower bound is immediate since functions in ${\mathcal C}_0$ are non-negative). We first derive an upper bound in the one-dimensional case. Let $f$ be a non-negative concave function on the real interval $[a, b]$. Let $x_0 \in (a, b)$ and suppose $f(x_0) = 1$. Let $x \in [a, x_0]$ and write $x_0 = \lambda x + (1 - \lambda)b$ for some $\lambda \in [0, 1]$. By concavity,
\[
1 = f(x_0) \geq \lambda f(x) + (1 - \lambda) f(b) \geq \lambda f(x).
\]
Thus
\[
f(x) \leq \frac{1}{\lambda} = \frac{b - x}{b - x_0} \leq \frac{b - a}{b - x_0}, \quad x \in [a, x_0].
\]
The case $x \in [x_0, b]$ can be handled similarly, and we get
\begin{equation} \label{eqn:concaveupperbound}
f(x) \leq \frac{b - a}{\min\{|x_0 - a|, |x_0 - b|\}}, \quad x \in [a, b].
\end{equation}
Now let $\Phi \in {\mathcal C}_0$. Applying \eqref{eqn:concaveupperbound} to the restrictions of $\Phi$ to line segments in $\Delta^{(n)}$ containing $p_0$, we get
\[
\Phi(p) \leq \frac{\text{diam}\left(\Delta^{(n)}\right)}{\text{dist}\left(p_0, \partial \Delta^{(n)}\right)}, \quad p \in \Delta^{(n)},
\]
where $\text{diam}(\Delta^{(n)})$ is the diameter of $\Delta^{(n)}$ and $\text{dist}(p_0, \partial \Delta^{(n)})$ is the distance from $p_0$ to the boundary of $\Delta^{(n)}$. This completes the proof of the lemma.
\qed\end{proof}

\begin{lemma} \label{lem:subdiffconvergence}
Let $(\pi, \Phi), (\pi^{(k)}, \Phi^{(k)}) \in {\mathcal{FG}}$, $k \geq 1$. Suppose $\Phi^{(k)}$ converges locally uniformly on $\Delta^{(n)}$ to $\Phi$. Let $p \in \Delta^{(n)}$ be a point at which $\Phi$ is differentiable. Then given $\varepsilon > 0$, there exists $\delta > 0$ and a positive integer $k_0$ such that $\|\pi^{(k)}(q) - \pi(p)\| < \varepsilon$ whenever $k \geq k_0$ and $q \in B(p, \delta)$. In particular, $\pi^{(k)}$ converges $m$-almost everywhere to $\pi$ as $k \rightarrow \infty$.
\end{lemma}
\begin{proof}
It is clear that $\log \Phi^{(k)}$ also converges locally uniformly to $\log \Phi$. We will use a well-known convergence result for the superdifferentials of concave functions, see \cite[Theorem 6.2.7]{HJL96}. Indeed, the proof of \cite[Theorem 6.2.7]{HJL96} implies a slightly stronger statement than the theorem. Namely, for any $p \in \Delta^{(n)}$ and any $\varepsilon > 0$, there exists a positive integer $k_0$ and $\delta > 0$ such that
\begin{equation} \label{eqn:claim1}
\begin{split}
\partial \log \Phi^{(k)}(q) &\subset \partial\log \Phi(p) + B(0, \varepsilon), \quad k \geq k_0, \quad q \in B(p, \delta), \\
\partial \log \Phi(q) &\subset \partial \log\Phi(p) + B(0, \varepsilon), \quad q \in B(p, \delta).
\end{split}
\end{equation}

Suppose $\Phi$ is differentiable at $p$. Then $\partial \log\Phi(p)$ is a singleton. By Lemma \ref{lem:superdiff}, there are measurable selections $\xi^{(k)}$ and $\xi$ of $\partial \log \Phi^{(k)}$ and $\partial \log \Phi$ respectively such that
\begin{equation*}
\begin{split}
\pi^{(k)}_i(q) &= q_i \left(\xi^{(k)}_i(q) + 1 - \sum_{j = 1}^n q_j\xi^{(k)}_j(q)\right), \\
\pi_i(q) &= q_i \left(\xi_i(q) + 1 - \sum_{j = 1}^n q_j\xi_j(q)\right), \\
\end{split}
\end{equation*}
for all $q \in \Delta^{(n)}$, $i = 1, ..., n$, and $k \geq 1$. 

For each $i = 1, ..., n$, consider the map $G_i$ defined by
\[
(q, \xi) \in \Delta^{(n)} \times T\Delta^{(n)} \mapsto q_i\left(\xi_i + 1 - \sum_{j = 1}^n q_j \xi_j\right).
\]
The map $G = (G_1, ..., G_n)$ is clearly jointly continuous. We have $\pi(q) = G(q, \xi(q))$ and $\pi^{(k)}(q) = G(q, \xi^{(k)}(q))$.

By \eqref{eqn:claim1}, for any $\varepsilon > 0$, there exists $k_0$ and $\delta > 0$ such that 
\begin{equation} \label{eqn:delta}
\|\xi^{(k)}(q) - \xi(p)\| < \varepsilon, \quad \|\xi(q) - \xi(p)\| < \varepsilon
\end{equation}
for all $k \geq k_0$ and $q \in B(p, \delta)$. The claim \eqref{eqn:claim1} follows from \eqref{eqn:delta} and the joint continuity of $G$ at $(q, \xi(q))$. The last statement follows since a finite concave function on $\Delta^{(n)}$ is differentiable $m$-almost everywhere \cite[Theorem 25.5]{R70}.
\qed\end{proof}

\begin{proof}[Proof of Theorem \ref{thm:optim}]
(i) The existence of an optimal solution will be proved by a compactness argument. Suppose $(\pi^{(k)}, \Phi^{(k)})$ is a maximizing sequence for \eqref{eqn:newoptim1}. By scaling, we may assume $\Phi^{(k)}(p_0) = 1$ where $p_0 \in \Delta^{(n)}$ is fixed. By Lemma \ref{lem:compactness}, we may replace it by a subsequence such that $\Phi^{(k)}$ converges locally uniformly on $\Delta^{(n)}$ to a positive concave function $\Phi$ on $\Delta^{(n)}$. By Lemma \ref{lem:superdiff}(ii), $\Phi$ generates a portfolio $\pi$.

\medskip
\noindent
{\it Case 1.} ${\Bbb P}$ is absolutely continuous.
By Lemma \ref{lem:subdiffconvergence}, $\pi^{(k)}$ converges $m$-almost everywhere to $\pi$. Let $T^{(k)}$ and $T$ be the L-divergences of $(\pi^{(k)}, \Phi^{(k)})$ and $(\pi, \Phi)$ respectively. Recall that ${\Bbb P}$ is supported on $K \times K$ where $K \subset \Delta^{(n)}$ is compact. For $x \in \overline{\Delta^{(n)}}$ and $p, q \in K$, we have
\begin{equation} \label{eqn:Tupperbound}
1 + \left\langle \frac{x}{p}, q - p \right\rangle = \sum_{i = 1}^n x_i \frac{q_i}{p_i} \leq \sum_{i   = 1}^n \frac{x_i}{p_i} \leq \frac{1}{\min_{p \in K, 1 \leq i \leq n} p_i}.
\end{equation}
Also $\Phi^{(k)} \rightarrow \Phi$ uniformly on $K$. Hence the family of L-divergences $\{T, T^{(1)}, T^{(2)}, ...\}$ is uniformly bounded on $K \times K$. By Lebesgue's dominated convergence theorem, we have
\[
\lim_{k \rightarrow \infty} \int T^{(k)}\left(q \mid p \right) \mathrm{d} {\Bbb P} = \int T\left(q \mid p \right) \mathrm{d} {\Bbb P}.
\]
Thus $(\pi, \Phi)$ is optimal.

\medskip
\noindent
{\it Case 2.} ${\Bbb P}$ is discrete and has masses at $(p(j), q(j))$. Since $\overline{\Delta^{(n)}}$ is compact, by a diagonal argument we can extract a further subsequence (still denoted by $\{(\pi^{(k)}, \Phi^{(k)})\}$) such that $\lim_{k \rightarrow \infty} \pi^{(k)}(p(j))$ exists for each $j$. Now we can redefine $\pi$ on $\{p(1), p(2), ...\}$ such that $\pi(p(j)) = \lim_{k \rightarrow \infty} \pi^{(k)}(p(j))$ for each $j$. Since we only modify $\pi$ at countably many points, $\pi$ is still Borel measurable. Now we may apply Lebesgue's dominated convergence theorem and conclude that $(\pi, \Phi)$ is optimal.

\medskip
\noindent
(ii) Suppose $(\pi^{(1)}, \Phi^{(1)})$ and $(\pi^{(2)}, \Phi^{(2)})$ are optimal solutions. Define $\pi = \frac{1}{2} \pi^{(1)} + \frac{1}{2} \pi^{(2)}$ which is generated by the geometric mean $\Phi  = \sqrt{\Phi^{(1)}\Phi^{(2)}}$ (Lemma \ref{lem:FGconvex}) . Also let $T$, $T^{(1)}$ and $T^{(2)}$ be the L-divergences of $(\pi, \Phi)$, $(\pi^{(1)}, \Phi^{(2)})$ and $(\pi^{(2)}, \Phi^{(2)})$ respectively. By concavity of the L-divergence (Lemma \ref{lem:Driftconcave}), we have
\begin{equation} \label{eqn:uniqueness2}
\int T\left(q \mid p \right) \mathrm{d} {\Bbb P} \geq \frac{1}{2} \left(\int T^{(1)}\left(q \mid p \right) \mathrm{d} {\Bbb P}  + \int T^{(2)}\left(q \mid p \right) \mathrm{d} {\Bbb P}\right).
\end{equation}
Hence $(\pi, \Phi)$ is also optimal. It follows from \eqref{eqn:uniqueness2} and the strict concavity of the logarithm that
\[
\left\langle \frac{\pi^{(1)}(p)}{p}, q - p \right\rangle = \left\langle\frac{\pi^{(2)}(p)}{p}, q - p \right\rangle
\]
for ${\Bbb P}$-almost all $(p, q)$.

If ${\Bbb P}$ is absolutely continuous and satisfies the support condition, then for $m$-almost all $p$ for which $f(p) > 0$, we have
\[
\left\langle \frac{\pi^{(1)}(p)}{p}, v \right\rangle = \left\langle\frac{\pi^{(2)}(p)}{p}, v \right\rangle
\]
for all tangent vectors $v$. This and the fact that $\pi^{(1)}(p), \pi^{(2)}(p) \in \overline{\Delta^{(n)}}$ imply that $\pi^{(1)}(p) = \pi^{(2)}(p)$ $m$-almost everywhere on $\{p: f(p) > 0\}$.
\qed\end{proof}

\begin{proof} [Proof of Theorem \ref{thm:optim2}]
By scaling, we may assume that $\widehat{\Phi}^{(N)}(p_0) = \Phi(p_0) = 1$ for all $N \geq 1$. By Lemma \ref{lem:compactness}, any subsequence of $\{\widehat{\Phi}^{(N)}\}$ has a further subsequence which converges locally uniformly to a positive concave function $\widehat{\Phi}$ on $\Delta^{(n)}$. Replacing $\{\widehat{\Phi}^{(N)}\}$ by such a convergent subsequence, we may assume that $\widehat{\Phi}^{(N)} \rightarrow \widehat{\Phi}$ locally uniformly on $\Delta^{(n)}$. Let $\widehat{\pi}$ be any portfolio generated by $\widehat{\Phi}$ (which exists by Lemma \ref{lem:superdiff}(ii)). We claim that $(\widehat{\pi}, \widehat{\Phi})$ is optimal and hence $\widehat{\pi} = \pi$ $m$-almost everywhere on $\{p: f(p) > 0\}$.

Let $\widehat{T}^{(N)}$, $\widehat{T}$ and $T$ be the L-divergences of $(\widehat{\pi}^{(N)}, \widehat{\Phi}^{(N)})$, $(\widehat{\pi}, \widehat{\Phi})$ and $(\pi, \Phi)$ respectively. By the optimality of $(\pi^{(N)}, \Phi^{(N)})$ for the measure ${\Bbb P}_N$, we have
\begin{equation} \label{eqn:finiteNoptimality}
\int \widehat{T}^{(N)}\left(q \mid p \right) \mathrm{d} {\Bbb P}_N \geq \int T\left(q \mid p \right) \mathrm{d} {\Bbb P}_N, \quad N \geq 1.
\end{equation}

We would like to let $N \rightarrow \infty$ in \eqref{eqn:finiteNoptimality}. The L-divergence $T\left( q \mid p \right)$ is clearly continuous on $K \times K$ (note that $K$ is compact). By the definition of weak convergence, we have
\[
\lim_{N \rightarrow \infty} \int T\left(q \mid p \right) \mathrm{d} {\Bbb P}_N = \int T\left(q \mid p\right) \mathrm{d} {\Bbb P}.
\]
Suppose we can prove that
\begin{equation} \label{eqn:bigclaim}
\lim_{N \rightarrow \infty} \int \widehat{T}^{(N)}\left(q \mid p \right) \mathrm{d} {\Bbb P}_N = \int \widehat{T}\left(q \mid p\right) \mathrm{d} {\Bbb P}.
\end{equation}
Then letting $N \rightarrow \infty$ in \eqref{eqn:finiteNoptimality}, we have
\[
\int \widehat{T}\left(q \mid p\right) \mathrm{d} {\Bbb P} \geq  \int T\left(q \mid p\right) \mathrm{d} {\Bbb P},
\]
so $(\widehat{\pi}, \widehat{\Phi})$ is optimal for the measure ${\Bbb P}$. Since ${\Bbb P}$ satisfies the support condition by assumption, by Theorem \ref{thm:optim}(ii) $\widehat{\pi}$ and $\pi$ are equal $m$-almost everywhere on $\{p: f(p) > 0\}$.

\medskip 

Thus we only need to prove \eqref{eqn:bigclaim}. Here the technicality lies in the fact that both the integrands and the measures change with $N$, so standard integral convergence theorems do not apply.

The main idea is to use the local uniform convergence property in Lemma \ref{lem:subdiffconvergence} and approximate the integrals in \eqref{eqn:bigclaim} by Riemann sums. Let $\varepsilon > 0$ be given. We will construct two partitions $\{A_k\}_{k = 0}^{k_0}$, $\{B_{\ell}\}_{\ell = 1}^{\ell_0}$ of $K$, points $p_k \in A_k$, $q_{\ell} \in B_{\ell}$ and a positive integer $N_0$ with the following properties:
\begin{enumerate}
\item[(i)] $A_k \times B_{\ell}$ is a ${\Bbb P}$-continuity set, i.e., ${\Bbb P}(\partial(A_k \times B_{\ell})) = 0$. Thus, by the Portmanteau theorem (see \cite{B09}), we have
\[
\lim_{N \rightarrow \infty} {\Bbb P}_N(A_k \times  B_{\ell}) = {\Bbb P}(A_k \times B_{\ell}).
\]
So for $N \geq N_0$ where $N_0$ is sufficiently large, we have
\[
\left|{\Bbb P}_N(A_k \times  B_{\ell}) - {\Bbb P}(A_k \times B_{\ell})\right| < \frac{\varepsilon}{k_0\ell_0}
\]
for all $k$, $\ell$.
\item[(ii)] ${\Bbb P}(A_0 \times K) < \varepsilon$ and ${\Bbb P}_N(A_0 \times K) < \varepsilon$ for $N \geq N_0$.
\item[(iii)] For $N \geq N_0$, $p \in A_k$, $q \in B_{\ell}$, $1 \leq k \leq k_0$ and $1 \leq \ell \leq \ell_0$, we have
\begin{equation*}
\left|\widehat{T}^{(N)}\left(q \mid p \right) - \widehat{T}\left(q_{\ell} \mid p_k\right)\right| < \varepsilon, \quad 
\left|\widehat{T}\left(q \mid p \right) - \widehat{T}\left(q_{\ell} \mid p_k\right)\right| < \varepsilon.
\end{equation*}
\item[(iv)] $\left|\log \widehat{\Phi}^{(N)}(p) - \log \widehat{\Phi}(p) \right| < \varepsilon$ for $p \in K$ and $N \geq N_0$. (This is immediate since $\widehat{\Phi}^{(N)}$ converges uniformly to $\widehat{\Phi}$ on $K$ and $\widehat{\Phi}$ is positive on $K$.)
\end{enumerate}

Suppose these objects have been constructed. Then for $N \geq N_0$ we can approximate the integrals as follows. By (ii) and (iii), we have
\begin{equation} \label{eqn:estimate1}
\begin{split}
& \left| \int \widehat{T}\left(q \mid p\right)\mathrm{d}{\Bbb P} - \sum_{\ell = 1}^{\ell_0} \sum_{k = 1}^{k_0} \widehat{T}\left(q_{\ell} \mid p_k\right) {\Bbb P}(A_k \times B_{\ell})\right| \\
&\leq \left| \int_{A_0 \times K} \widehat{T}\left(q \mid p\right) \mathrm{d}{\Bbb P} \right| + \sum_{\ell = 1}^{\ell_0} \sum_{k = 1}^{k_0} \int_{A_k \times B_{\ell}} \left|\widehat{T}\left(q \mid p \right) - \widehat{T}\left(q_{\ell} \mid p_k\right)\right| \mathrm{d} {\Bbb P} \\
&\leq \varepsilon \max_{p, q \in K} \widehat{T}\left(q \mid p\right) + \varepsilon.
\end{split}
\end{equation}
Similarly, we have
\begin{equation} \label{eqn:estimate2}
\begin{split}
& \left|  \int \widehat{T}^{(N)}\left(q \mid p\right)\mathrm{d}{\Bbb P}_N - \sum_{\ell = 1}^{\ell_0} \sum_{k = 1}^{k_0} \widehat{T}\left(q_{\ell} \mid p_k\right) {\Bbb P}_N(A_k \times B_{\ell})  \right|  \\
&\leq   \varepsilon \max_{p, q \in K} \widehat{T}^{(N)}\left(q \mid p\right) + \varepsilon.
\end{split}
\end{equation}

By \eqref{eqn:Tupperbound} and uniform convergence of $\{\widehat{\Phi}^{(N)}\}$ on $K$, we can bound $\max_{p, q \in K} \widehat{T}\left(q \mid p\right)$ and $ \max_{p, q \in K} \widehat{T}^{(N)}\left(q \mid p\right)$ by a constant $C$. Using (i) and (iii), we get
\begin{equation} \label{eqn:estimate3}
\begin{split}
& \left| \sum_{k, \ell} \widehat{T} \left(q_{\ell} \mid p_k \right) {\Bbb P}_N(A_k \times B_{\ell}) - 
\sum_{k, \ell} \widehat{T} \left(q_{\ell} \mid p_k \right) {\Bbb P}(A_k \times B_{\ell}) \right| \\
&\leq \sum_{k, \ell} \widehat{T}\left(q_{\ell} \mid p_k\right) \left| {\Bbb P}_N(A_k \times B_{\ell}) - {\Bbb P}(A_k \times B_{\ell})\right| \\
&\leq k_0\ell_0C\frac{\varepsilon}{k_0\ell_0} = C\varepsilon.
\end{split}
\end{equation}
Combining \eqref{eqn:estimate1}, \eqref{eqn:estimate2} and \eqref{eqn:estimate3}, we have the estimate
\[
\left|\int \widehat{T}^{(N)}\left(q \mid p\right) \mathrm{d}{\Bbb P}_N - \int \widehat{T}\left(q \mid p\right) \mathrm{d}{\Bbb P}\right| \leq (3C + 2)\varepsilon, \quad N \geq N_0,
\]
and so \eqref{eqn:bigclaim} holds.

\medskip

It remains to construct the sets $\{A_k\}$, $\{B_{\ell}\}$, the points $p_k$, $q_{\ell}$ and $N_0$ satisfying (i)-(iv). Before we begin, we note the fact that the boundary of any convex subset of $\Delta^{(n)}$ has $m$-measure zero \cite[Theorem 1]{L86}. Let $\varepsilon > 0$ be given. By \cite[Theorem 10.6]{R70}, the family $\{\widehat{\Phi}, \widehat{\Phi}^{(1)}, \widehat{\Phi}^{(2)}, ...\}$ is uniformly Lipschitz on $K$. Also, it is not difficult to verify that there exists a constant $L > 0$ so that
\[
\left|\log \left(1 + \left \langle \frac{x}{p}, q - p \right\rangle\right) - \log \left(1 + \left \langle \frac{x}{p'}, q' - p' \right\rangle\right)\right| \leq L\left(\|p - p'\| + \|q - q'\|\right)
\]
for all $x \in \overline{\Delta^{(n)}}$ and $p, p', q, q' \in K$. It follows that the family of L-divergences $\{\widehat{T}, \widehat{T}^{(1)}, \widehat{T}^{(2)} ...\}$ is uniformly Lipschitz on $K \times K$. Thus there exists $\delta_0 > 0$ such that if $p, p', q, q' \in \Delta^{(n)}$, then
\begin{equation} \label{eqn:energyconverge}
\left|\widehat{T}^{(N)}\left(q' \mid p' \right) - \widehat{T}^{(N)}\left(q \mid p \right)\right| < \frac{\varepsilon}{2} \quad \text{and} \quad
\left|\widehat{T}\left(q' \mid p' \right) - \widehat{T}\left(q \mid p \right)\right| < \varepsilon
\end{equation}
whenever $\|q - q'\| < \delta_0$, $\|p - p'\| < \delta_0$.

Let $D$ be the set of points in $K$ at which $\widehat{\Phi}$ is differentiable. Then $K \setminus D$ has $m$-measure zero by \cite[Theorem 25.5]{R70}. Let $\varepsilon' > 0$ be arbitrary. By Lemma \ref{lem:subdiffconvergence}, for each $p \in D$ there exists $0 < \delta(p) \leq \delta_0$ and a positive integer $N_0(p)$ such that $\left\|\widehat{\pi}^N(q) - \widehat{\pi}(p)\right\| < \varepsilon'$ for all $N \geq N_0(p)$ and $q \in B(p, \delta(p))$. 

Since $K$ is compact, it is separable, and so is $D$ as a subset of $K$. The collection $\{B(p, \delta(p))\}_{p \in D}$ forms an open cover of $D$ and hence there exists a countable subcover. By the continuity of measure, for any $\eta > 0$ there exists $p_1, ..., p_{j_0} \in D$ such that
\[
m(A_0) < \eta, \quad A_0 := K \setminus \bigcup_{j = 1}^{j_0} B(p_j, \delta(p_j)),
\]
Since $\partial A_0 \subset \partial K \cup \bigcup_j \partial B(p_j, \delta(p_j))$, $\partial (A_0 \times K)$ has $m$-measure zero and hence $A_0 \times K$ is a ${\Bbb P}$-continuity set. Since ${\Bbb P}$ is absolutely continuous, choosing $\eta > 0$ sufficiently small we have
\[
{\Bbb P}(A_0 \times K) < \varepsilon,
\]
and by weak convergence we have ${\Bbb P}_N(A_0 \times K) < \varepsilon$ for $N$ sufficiently large, so (ii) holds. Let $A_1 = B(p_1, \delta(p_1)) \cap K$ and define $A_k = \{p_k\} \cup (B(p_k, \delta(p_k)) \cap K) \setminus (A_1 \cup \cdots \cup A_{k-1})$, $j = 2, ..., k_0$. If $N \geq \max_{1 \leq k \leq k_0} N_0(p_k)$, we have
\begin{equation} \label{eqn:estimate4}
\left\|\widehat{\pi}^N(p) - \widehat{\pi}(p_k)\right\| < \varepsilon', \quad p \in A_k, \quad k = 1, ..., k_0.
\end{equation}
Next choose $q_1, ..., q_{\ell_0} \in K$ such that $K \subset \bigcup_{\ell = 1}^{\ell_0} B(q_{\ell}, \delta_0)$. Define $B_1 = B(q_1, \delta_0) \cap K$ and $B_{\ell} = \{q_{\ell}\} \cup (B(q_{\ell}, \delta_0) \cap K) \setminus (B_1 \cup \cdots \cup B_{\ell-1})$, $j = 2, ..., \ell_0$. Again it is clear that $\partial(A_k \times B_{\ell})$ has $m$-measure zero and is a ${\Bbb P}$-continuity set. So (i) holds for $N$ sufficiently large. Finally, if we choose $\varepsilon' > 0$ small enough in \eqref{eqn:estimate4}, we have
\[
\left|\widehat{T}^{(N)}\left( q \mid p \right) - \widehat{T}\left( q \mid p_k \right)\right| < \frac{\varepsilon}{2}, \quad p \in B(p_k, \delta_0), \quad q \in \Delta^{(n)}
\]
for $N$ sufficiently large. This and \eqref{eqn:energyconverge} imply (iii) and the proof of Theorem \ref{thm:optim2} is complete.
\qed\end{proof}

\bibliographystyle{amsalpha}
\bibliography{infogeo}

\end{document}